\documentclass[11pt]{article}

\usepackage{url}
\usepackage{mathtools}
\usepackage{amsmath}
\usepackage{amssymb}
\usepackage{amsthm}
\usepackage{empheq}
\usepackage{latexsym}
\usepackage{enumitem}
\usepackage{eurosym}
\usepackage{dsfont}
\usepackage{appendix}
\usepackage{color} 
\usepackage[unicode]{hyperref}
\usepackage{frcursive}
\usepackage[utf8]{inputenc}
\usepackage[T1]{fontenc}
\usepackage{geometry}
\usepackage{multirow}
\usepackage[colorinlistoftodos]{todonotes}
\usepackage{lmodern}
\usepackage{anyfontsize}
\usepackage{stmaryrd}
\usepackage{natbib}
\usepackage{cleveref}

\usepackage{ulem}
\usepackage{amsbsy}

\usepackage{comment}

\setcitestyle{numbers,open={[},close={]}}

\definecolor{red}{rgb}{0.7,0.15,0.15}
\definecolor{green}{rgb}{0,0.5,0}
\definecolor{blue}{rgb}{0,0,0.7}
\hypersetup{colorlinks, linkcolor={red},citecolor={green}, urlcolor={blue}}
			
\makeatletter \@addtoreset{equation}{section}

\newtheorem{theorem}{Theorem}[section]

\newtheorem{proposition}[theorem]{Proposition}

\newtheorem{definition}[theorem]{Definition}
\newtheorem{remark}[theorem]{Remark}

\def\no{\noindent}
\def\beq{\begin{eqnarray}}
\def\eeq{\end{eqnarray}}
\def\be*{\begin{eqnarray*}}
\def\ee*{\end{eqnarray*}}

\def \E{\mathbb{E}}
\def \F{\mathbb{F}}

\def \L{\mathbb{L}}
\def \M{\mathbb{M}}
\def \N{\mathbb{N}}

\def \P{\mathbb{P}}

\def \R{\mathbb{R}}
\def \S{\mathbb{S}}

\def \X{\mathbb{X}}

\def \Pr{\mathrm{P}}

\def\Ac{{\cal A}}
\def\Bc{{\cal B}}
\def\Cc{{\cal C}}

\def\Fc{{\cal F}}
\def\Gc{{\cal G}}

\def\Jc{{\cal J}}

\def\Lc{{\cal L}}
\def\Mc{{\cal M}}

\def\Pc{{\cal P}}

\def\Rc{{\cal R}}

\def\Vc{{\cal V}}
\def\Wc{{\cal W}}

\def\Xb{{\bar X}}

\def\x{\times}
\def\eps{\varepsilon}

\def\Om{\Omega}

\def\0{\mathbf{0}}

\def \mub{\overline{\mu}}

\def \muh{\widehat{\mu}}

\def\normeL2#1{\left\|{#1}\right\|_{L^2}}

\def\alphah{\widehat \alpha}

\def\Xh{\widehat X}
\def\Wh{\widehat W}

\def\Drm{{\rm D}}
\def\Dbf{\mathbf{D}}
\def\Gbf{\mathbf{G}}

\def \Lim{\displaystyle\lim}

\DeclareMathOperator*{\esup}{ess\,sup}

\def \alphab {\boldsymbol{\alpha}}
\def \betab {\boldsymbol{\beta}}

\def \Xb{\overline{X}}

\def \Xbb{\mathbf{X}}

\def \Ybb{\mathbf{Y}}

\def \Qr{\mathrm{Q}}

\def \1{\mathds{1}}
\def \d{{\rm d}}
\def \ep{\hbox{ }\hfill$\Box$}

\DeclareUnicodeCharacter{014D}{\=o}
 \title{\bf On Approximate Nash Equilibria  \\ in Mean Field Games}

\author{
Mao Fabrice Djete\footnote{\'Ecole Polytechnique Paris, Centre de Math\'ematiques Appliqu\'ees. This work benefits from the financial support of the Chairs {\it Financial Risk} and {\it Finance and Sustainable Development. Email: mao-fabrice.djete@polytechnique.edu}} 
\and Nizar Touzi\thanks{NYU, Tandon School of Engineering, USA. This author is partially supported by NSF grant \#DMS-2508581. Email: nizar.touzi@nyu.edu} 
}
             \date{\today}

\begin{document}

\maketitle
 
\begin{abstract}
In the context of large population symmetric games, approximate Nash equilibria are introduced through equilibrium solutions of the corresponding mean field game in the sense that the individual gain from optimal unilateral deviation under such strategies converges to zero in the large  population size asymptotic. We show that these strategies satisfy an $\L^\infty$ notion of approximate Nash equilibrium which guarantees that the individual gain from optimal unilateral deviation is small uniformly among players and uniformly on their initial characteristics. We establish these results in the context of static models and in the dynamic continuous time setting, and we cover situations where the agents' criteria depend on the conditional law of the controlled state process.  
\end{abstract}

\vspace{3mm}
\no{\bf Keywords.} Approximate Nash equilibrium, mean field McKean-Vlasov stochastic differential equation, mean field game. 

\vspace{3mm}
\no{\bf MSC2010.} 60K35, 60H30, 91A13, 91A23, 91B30.

\section{Introduction}\label{sec:intro}

Games with a large number of strategic agents arise naturally in economics, engineering, and social sciences, including models of financial engineering, congestion, auctions, communication networks, and population dynamics. While Nash equilibrium provides a canonical solution concept, its computation and characterization quickly become intractable as the number of players grows. Mean field game (MFG) theory addresses this challenge by studying the limiting behavior of symmetric finite player games as the population size $n$ increases to infinity, replacing direct strategic interaction with interaction through an aggregate distribution, or mean field. Since the seminal work of  \citeauthor*{lasry2006jeux} \cite{lasry2006jeux,lasry2006jeux2,lasry2007mean} and \citeauthor*{huang2006large} \cite{huang2006large}, MFG theory has become a central tool for analyzing large-population strategic systems. In the mean field framework, each agent solves an individual optimization problem against an exogenously given flow of population states or actions, while equilibrium requires consistency between the optimal policy and the induced population distribution. The resulting solution concept, called mean field equilibrium (MFE), is substantially simpler than Nash equilibrium of the finite-player game. A central theoretical question is therefore whether, and in what sense, an MFE provides a good approximation to a Nash equilibrium when the number of players is large but finite.

A large body of work establishes that strategies derived from an MFE yield an approximate Nash equilibrium -~also called $\varepsilon-$Nash equilibrium~- in the sense that, in the corresponding $n$--population size game, the individual gain from optimal unilateral deviation under such MFE based strategies vanishes as the size of the population $n$ converges to infinity. In continuous-time stochastic differential games, this program was first carried out for general nonlinear large-population stochastic dynamic games by \citeauthor*{huang2006large} \cite{huang2006large} who showed that this gain from optimal unilateral deviation is of order $n^{-\frac12}$. We also refer to \citeauthor*{cardaliaguet2010notes} \cite{cardaliaguet2010notes} for the first order context. \citeauthor*{carmona2013probabilistic} \cite{carmona2013probabilistic} obtained this result by following the Pontryagin maximum principle approach to the representative agent control problem. 
These results were specialized and sharpened for cost-coupled Linear-Quadratic-Gaussian models in \citeauthor*{Bensoussan2014LinearQuadraticMF} \cite{Bensoussan2014LinearQuadraticMF} and \cite{huang2007large}. The corresponding result in the context of closed loop strategies was obtained by \citeauthor*{cardaliaguet2015master} \cite{cardaliaguet2015master} by using PDE techniques at the level of the master equation on the torus. Subsequent extensions cover more general dynamics and cost structures; see, among others, \citeauthor*{carmona2018probabilisticI}  \cite{carmona2018probabilisticI}. These results rely on propagation-of-chaos arguments and stability estimates showing that the impact of a single agent's deviation on the empirical distribution vanishes asymptotically.

The notion of an $\varepsilon-$Nash equilibrium in static games was introduced explicitly by \citeauthor*{RADNER1980136} \cite{RADNER1980136} who defined $\varepsilon-$equilibria as strategy profiles in which no player can gain more than $\varepsilon$ by a unilateral deviation. Radner's motivation was to model approximate rationality and robustness in large economic systems, particularly oligopolistic markets, where exact best responses may be unrealistic or analytically fragile. This concept provided a formal relaxation of Nash equilibrium while retaining clear strategic meaning. Subsequent work in static game theory adopted $\varepsilon-$Nash equilibria as a tool for studying stability and approximation of equilibria, especially in large games where individual players have limited influence. In parallel, the literature on large games and continuum limits - building on earlier nonatomic game theory of e.g., \citeauthor*{Schmeidler1973EquilibriumPO}  \cite{Schmeidler1973EquilibriumPO} - recognized that equilibria of nonatomic games typically correspond, when interpreted in finite-player settings, to approximate equilibria rather than exact Nash equilibria, even though $\varepsilon-$Nash equilibria were not always introduced explicitly in that early work. Over time, $\varepsilon-$Nash equilibria became the standard language for formalizing this approximation viewpoint. More recently, $\varepsilon-$Nash equilibria have also played a central role in algorithmic game theory, where exact Nash equilibria may be computationally intractable; see, for instance, \citeauthor*{Lipton2003PlayingLG} \cite{Lipton2003PlayingLG} on approximate Nash equilibria in large normal-form games, and \citeauthor*{Daskalakis2006TheCO}  \cite{Daskalakis2006TheCO} on the complexity of computing approximate equilibria. Collectively, this literature establishes $\varepsilon-$Nash equilibria as a fundamental solution concept for static games, serving both as a robustness notion in economic modeling and as a practical approximation tool in large and complex strategic environments.

A common feature of nearly all these approximation results is that convergence of the gain from optimal deviating is established in an averaged sense, typically in expectation or in probability. More precisely, the standard definition of approximate Nash equilibrium bounds the expected gain from unilateral deviation, leading naturally to $\L^1-$type estimates with respect to the underlying probability measure governing the state dynamics. In continuous-time stochastic models, error bounds are usually derived for expected costs conditional on the common filtration, see e.g. \citeauthor*{carmona2018probabilisticI}  \cite{carmona2018probabilisticI},  \citeauthor*{huang2006large} \cite{huang2006large}, while in discrete-time models the approximation holds in expectation over the empirical distribution of states or actions, see e.g. \citeauthor*{ADLAKHA2015269}  \cite{ADLAKHA2015269}. Although such results are sufficient for many purposes, they do not provide uniform or worst-case guarantees on the approximation quality, and possibly allow for rare but potentially large deviations of realized payoffs from their mean field limits.

This paper contributes to the literature by strengthening the mode of convergence of the gain from optimally deviating. Our focus is on $\L^\infty-$convergence results of the Nash equilibrium error, yielding uniform bounds over the relevant state space. Of course, when the underlying state space is compact, this does not differ from the previously established results. Our stronger notion of approximation provides robust guarantees that the MFE-based strategies provide equal benefit for all agents regardless of their individual characteristics and uniformly across realizations and initial conditions. To the best of our knowledge, such $\L^\infty$ bounds are largely absent from the existing MFG approximation literature, both in continuous and discrete time, and require new analytical techniques beyond standard propagation-of-chaos arguments. This refinement strengthens the economic and algorithmic interpretation of MFE and broadens its applicability in settings where worst-case performance guarantees are essential.

Throughout this paper, we consider the general situation where the individual player's criterion is a function of the conditional law of its state, conditional on its own initial characteristics. We first derive our main approximate Nash equilibrium results in discrete time -~actually in a one-period model~- which raises some special difficulties due to some features related to the discrete time setting. We next cover the continuous time setting in the simple uncontrolled diffusion setting, and open-loop finite population game. Due to the nature of the agents' performance criterion, we also complement our analysis with the corresponding existence results of MFE. 

Our main findings are established under several refinements of the standard Lipschitz condition on the coefficients driving the dynamics of the game problem:
\begin{itemize}
\item We derive an $\L^\infty-$approximate Nash result under the strong $\Wc_0-$Lipschitz condition; this is our first $\varepsilon-$Nash equilibrium result in the $\L^\infty-$sense. 
\item Next, under the weaker, and more standard, $\Wc_p-$Lipschitz condition, we derive a new $\L^\infty-$approximate Nash result for a large portion of the population which increases as an appropriate power $n^\delta$ of the population size. For this reason, we call $(\varepsilon,\delta)-$Nash equilibrium this new notion of approximate Nash-equilibria.
\end{itemize}
The paper is organized as follows. Section \ref{sec:pbfm} introduces the static symmetric finite population game problem and the corresponding mean field limit, and reports our main approximate Nash equilibrium results. We next derive in Section \ref{sec:continuous} similar results in the uncontrolled diffusion continuous time setting. We finally report the existence results of mean field equilibrium in the static and continuous settings in Sections \ref{sect:MFGStatic} and \ref{sect:MFGContinuous}, respectively.

\vspace{3mm}
\noindent {\bf Notations}\quad 
For a generic Polish space $E$, we denote $\Bc_E$ the collection of all Borel subsets of $E$, and by $\Pc(E)$ the collection of all probability measures $m:A\in\Bc_E\longmapsto m[A]\in[0,1]$. We denote similarly by $m(f) := \int_{E}f\d m$, for any $m$--integrable function $f:E\longrightarrow\R$. 

For $q\in \{0\} \cup [1,\infty)$, the subset $\Pc_{q}(E)\subset\Pc(E)$ consists of all measures $m$ with finite $q$-th moment $\int_E|x|^qm(\d x)<\infty$, and we denote by $\Wc_q$ the corresponding $q-$Wasserstein distance:
\[
\Wc_q(\gamma,\gamma')^{1 \vee q} 
:=
\inf_{\pi \in \Pi(\gamma, \gamma')}  
\int_{E \x E} c_q(e-e') ~\pi( \mathrm{d}e, \mathrm{d}e'),
~c_q(x):=|x|^q\1_{\{ q \ge 1\}}+1\wedge |x|\1_{\{q=0\}},
\]
with $\Pi(\gamma, \gamma')$ the collection of all coupling probability measures with marginals $\gamma$ and $\gamma'$. Finally, $C(E)$ is the space of continuous real-valued functions on $E$, and $C_b(E)$ denotes the corresponding subset of bounded functions. 

\section{One-period approximate Nash equilibria}\label{sec:pbfm}

Let $(\Omega,\Fc,\P)$ be  probability space.
Throughout this section, $\S$ is a Polish space that will stand as the state space. We denote 
\begin{eqnarray*}
\bar\S:=\S\times[0,1]
&\mbox{and}&
\X:=\S\times\S
\end{eqnarray*}
 the randomized state space and the state space a 2-step dynamic process $X=(X_0,X_1)$ on $(\Omega,\Fc)$, respectively. Here, the $[0,1]-$component of $\bar\S$ will serve as the support of an independent uniform distribution ${\rm U}_{[0,1]}$. We shall consider control strategies taking values in a Polish space $A$  with bounded distance $\rho_{_{\!\!A}}$.

\subsection{Finite population Nash equilibrium}

We consider a finite population of $n$ agents whose actions can be chosen from the set $\Ac_n:=\L^0(\bar\S^n,A)$ of all Borel measurable maps from $\bar\S^n$ to a Polish space $A$ . We denote by $\boldsymbol{\alpha}:=(\alpha^1,\cdots,\alpha^n)\in(\Ac_n)^n$ the vector of actions of the $n$ agents. As standard in game theory, we use the standard notation  $\boldsymbol{\alpha}^{-i}\in(\Ac_n)^{n-1}$ for the vector of $n-1$ actions excluding that of the the $i-$th agent, and $\boldsymbol{\alpha}^{-i}\!\oplus\!\beta$ for the vector of $n$ actions where $\alpha^i$ is replaced by $\beta\in\Ac_n$.

\medskip
Let $\mathbf{X}_0^n:=(X^i_0)_{i\le n}$ and $\boldsymbol{\alpha}:=(\alpha^1,\cdots,\alpha^n) \in (\Ac_n)^n$ be a family of r.v. representing the initial states and the actions of a finite population of $n$ agents. Let $(\xi,\eta)$ be a pair of independent random variables with $\Lc_\xi={\rm U}_{[0,1]}$, and consider a family $(\eta^i, \xi^i)_{i \ge 1}$ of independent copies of $(\eta,\xi)$ and also independent of $(X^i_0)_{i \ge 1}$. We shall consider in \eqref{eq:X} below the evolution of the underlying state to $(X^i_1)_{i\le n}$ by means of some strategy $\boldsymbol{\alpha}$, and we denote $X^i=(X^i_0,X^i_1)$ the resulting state process. We introduce additional notations for the randomized initial states and the pairs state-control:
$$
\bar X^i_0:=(X^i_0,\xi^i)\in\bar\S
~~\mbox{and}~~
\Theta^i:=(X^i,\alpha^i(\boldsymbol{\bar X}^n_0))\in\X\times A,~~\mbox{a.s.}
$$
where $\boldsymbol{\bar X}^n_0:=(\bar X^1_0,\ldots,\bar X^n_0)$, $\boldsymbol{\Theta}^n:=(\Theta^1,\ldots,\Theta^n)$.

\medskip
The population state process $\Xbb^{n,\boldsymbol{\alpha}}=(\Xbb^n_0,\Xbb^{n,\boldsymbol{\alpha}}_1):=(X^{n,\boldsymbol{\alpha},1},\cdots,X^{n,\boldsymbol{\alpha},n})$ is defined by the state equation for each $X^{n,\boldsymbol{\alpha},i}$, $i\!=\!1,\ldots,n$:
\begin{eqnarray}\label{eq:X}
    X^{n,\boldsymbol{\alpha},i}_1
    =
    F^{\eta^i }_{X^i_0} \big(\mu^n_{\boldsymbol{\Theta}^n}, 
                                                       \alpha^{i}(\boldsymbol{\bar X}^n_0)
                                 \big), 
    &\mbox{where}&
    \mu^n_{\boldsymbol{\Theta}^n}
    \!:=\!
    \frac{1}{n} \sum_{i=1}^n \delta_{\Theta^{n,i}},
\end{eqnarray}
and $F: \R\x\S\x \Pc(\X \x A) \x A  \longrightarrow \S$ is a Borel map $(e,x_0,m,a)\longmapsto F^e_{x_0}(m,a)$ satisfying appropriate conditions, see Assumption ${\rm H}_p$ below. 

The dynamics \eqref{eq:X} couples the final value of the state $\mathbf{X}$ through the dependence of $F$ on the environment, represented by the empirical measure $\mu^n_{\boldsymbol{\Theta}^n}$.

\begin{remark}{\rm
Notice that $\Xbb^{n,\alphab}_1$ appears on both sides of \eqref{eq:X} as it is contained in the right hand side through the empirical measure $\mu^n_{\boldsymbol{\Theta}^n}$. In the present one-period setting this is the only possible source of interaction of the particles dynamics. For this reason, $\Xbb^{n,\alphab}$ does not fit into the standard class of Markov decision process in the literature, and we shall introduce an appropriate assumption to ensure the well--posedness of $\Xbb^{n,\alphab}_1$ through \eqref{eq:X}. While this feature needs to be addressed explicitly in discrete time, we observe that it raises no issues in the continuous time formulation of the problem, due to the infinitesimal time step.
} 
\end{remark}

\medskip
The performance criterion of each agent $i$ is defined as a functional of the conditional law $\Lc_{X^{n,\boldsymbol{\alpha},i}_1}^{|X^i_0}:=\P\circ(X^{n,\boldsymbol{\alpha},i}_1|X^i_0)^{-1}$: 
\begin{align*}
    J_i^n(X^i_0,\boldsymbol{\alpha})
    := 
    U \big( \Lc_{X^{n,\boldsymbol{\alpha},i}_1}^{|X^i_0} \big)
    ~~\mbox{and}~~
    V_i^n(X^i_0,\boldsymbol{\alpha}^{-i})
    :=
    \esup_{\alpha^i\in\Ac_n}J_i^n(X^i_0,\boldsymbol{\alpha}),
\end{align*}
where $U:\Pc(\S) \longrightarrow \R \cup \{ - \infty\}$ is a given Borel map satisfying appropriate continuity conditions, see Assumption ${\rm H}_p$ below. 

For any strategy $\boldsymbol{\alpha}\in(\Ac_n)^n$, we introduce the agents deviation gains:
\begin{eqnarray}\label{Dn}
\Drm_n^i[\boldsymbol{\alpha}]
:=
V_i^n(X^i_0,\boldsymbol{\alpha}^{-i})-J_i^n(X^i_0,\boldsymbol{\alpha}),
~i=1\ldots,n,
\end{eqnarray}
representing the maximum gain that Agent $i$ would realize by deviating from the strategy $\alpha^i$ given the environment strategy $\boldsymbol{\alpha}^{-i}$. Clearly, $\Drm_n^i[\boldsymbol{\alpha}]$ is non-negative, and we introduce
\begin{eqnarray}\label{Dn}
\overline{\Dbf}^\delta_n[\boldsymbol{\alpha}]
:=
\max_{i\le n}\Dbf_n^i [\boldsymbol{\alpha}]\1_{\{X^i_0\in B(n^\delta)\}},
~~\delta\in(0,\infty].
\end{eqnarray}
where $B(R)$ denotes the centered ball in $\R^d$ with radius $R>0$. Here $\overline{\Dbf}^\infty_n$ is the population maximum deviation gain and for $\delta\in(0,\infty)$, the truncated population maximum deviation $\overline{\Dbf}^\delta_n$ ignores the portion of the population outside an expanding ball of diverging radius. The following definition recalls the standard notion of Nash equilibrium and $\varepsilon-$Nash equilibrium in the context of a finite population game, and introduces additional notions of approximate Nash equilibrium that will be the main focus of the present paper.

\begin{definition}\label{def:Nash}
For fixed $n\ge 1$ and $\varepsilon,\delta>0$, a strategy $\boldsymbol{\alpha}\in(\Ac_n)^n$ is called
\\
{\rm (NE)} a Nash equilibrium for $n-$player game if $\overline{\Dbf}^\infty_n[\boldsymbol{\alpha}]=0$;
\\
{\rm (NE$_\varepsilon$)} an $\varepsilon-$Nash equilibrium for $n-$player game if $\overline{\Dbf}^\infty_n[\boldsymbol{\alpha}]\le\varepsilon$.
\\
{\rm (NE$^\delta_\varepsilon$)} an $(\varepsilon,\delta)-$Nash equilibrium for $n-$player game if $\overline{\Dbf}^\delta_n[\boldsymbol{\alpha}]\le\varepsilon$.
\end{definition}

\begin{remark}\label{rem:approx_Nash_uniform}
{\rm Definition~\ref{def:Nash} introduces a hierarchy of equilibrium notions tailored to different levels of approximation in large population games. By construction, a Nash equilibrium corresponds to the strongest requirement, enforcing the absence of profitable unilateral deviations for every agent in the population. The notion of $\varepsilon$--Nash equilibrium relaxes this condition by allowing deviation gains of order at most $\varepsilon$, uniformly over all agents $i\le n$.

The $(\varepsilon,\delta)$--Nash equilibrium further weakens the equilibrium requirement by restricting the uniform control of deviation gains to agents whose initial states belong to the expanding ball $B(n^\delta)$. Importantly, this control remains uniform with respect to the agent index \(i\); the relaxation concerns only the region of the state space on which the equilibrium condition is enforced, and not the uniformity of the deviation bound across the relevant subpopulation.

This notion is particularly well suited to asymptotic regimes in which the initial states $\{X^i_0\}_{i\le n}$ may have unbounded support. In such settings, uniform estimates can typically be established on compact subsets of the state space, while the contribution of agents outside $B(n^\delta)$ is handled via tail estimates or moment conditions. As $n\to\infty$, choosing $\delta$ sufficiently large allows the truncated region to capture an overwhelming fraction of the support of the initial distribution, so that $(\varepsilon,\delta)$--Nash equilibria provide a meaningful approximation of full $\varepsilon$--Nash equilibria in the large population limit. 

Finally, this notion is well adapted to obtaining $\L^\infty$--type approximations of Nash equilibria, in contrast with the $\L^1$ or averaged notions more commonly considered in the literature; see {\rm Theorem~\ref{thm:approx_nash_normal}} below for a precise illustration of this feature.
}\end{remark}
 
\subsection{The mean field game}

Our main focus in this paper is to derive approximate Nash equilibria for large population games, i.e. in the infinite population asymptotics $n\to\infty$. We follow the standard method in the literature on large population symmetric games which defines such approximate Nash equilibria $\boldsymbol{\alpha}^n$ through a solution of the limiting mean field problem that we now introduce:
\begin{itemize}
\item Let $\alpha\in\Ac:=\Ac_1=\L^0(\bar\S,A)$, the set of Borel measurable maps from $\bar\S$ to $A$;
\item Let $X_0$, $\xi$ and $\eta$ be independent r.v. with $\Lc_{X_0}=\mu_0$, $\Lc_\xi={\rm U}_{[0,1]}$, and denote the randomized initial state $\bar X_0:=(X_0,\xi)$;
\item Given $\mu\in \Pc(\X \x A)$ with first marginal $\mu_0=\mu\circ X_0^{-1}$, we introduce the $\X-$valued process $X^{\mu,\alpha}=(X_0,X^{\mu,\alpha}_1)$ by the transition dynamics:
\begin{align*}
    X^{\mu,\alpha}_1= F_{X_0}^\eta\big(\mu, \alpha(\bar X_0) \big).
\end{align*}
\end{itemize}
Using the Borel map $U:\Pc(\S) \longrightarrow \R \cup \{- \infty \}$ as in the finite population setting, we consider the representative agent's performance map
\begin{align*}
    J(\mu,\alpha) := \E \big[ U \big( \Lc_{X^{\mu,\alpha}_1}^{| X_0} \big)\big],
    ~~\mbox{for all}~~
    \mu\in\Pc(\X\times A),~\mbox{and}~\alpha\in \Ac,
\end{align*}
where $\Lc_{X^{\mu,\alpha}_1}^{| X_0}:=\P\circ(X^{\mu,\alpha}_1| X_0)^{-1}$, and the expectation is well defined by the convention $\infty-\infty=-\infty$.

\begin{definition}
A measure $\hat\mu \in \Pc(\X \x A)$ is an MFG solution associated to $\hat{\alpha} \in \Ac$ if: 
    \begin{enumerate}
        \item[{\rm (IO)}] given $\hat\mu$, the strategy $\hat{\alpha}$ solves the individual optimality $J(\hat\mu,\hat{\alpha})=\max_{\alpha \in \Ac} J(\hat\mu,\alpha)$;
        \item[{\rm (FP)}] and $\hat\mu$ satisfies the fixed point condition $\P\circ(X^{\hat\mu,\hat{\alpha}},\, \alphah (\bar X_0))^{-1}=\hat\mu$.
    \end{enumerate}
\end{definition}

Since the control $\alpha$ depends on $X_0$, notice that Condition (IO) in the last definition of the MFG means that 
$$
U \big( \Lc_{X^{\hat\mu,\hat{\alpha}(\bar X_0)}_1}^{| X_0}\big)
=
V(X_0,\hat\mu),
~
\mbox{with}~
V(X_0,\mu)
:=
\esup_{a\in \L^0(A)}
U \big( \Lc_{X^{\mu,a}_1}^{| X_0}\big)
,~\mu_0-\mbox{a.s.} 
$$
which expresses the standard individual optimization in the present mean field notion of Nash equilibrium.

Throughout this section, we impose appropriate conditions on the nonlinearity $F$ involving a parameter $p\ge 0$ which will be specified in our subsequent results.

\newtheorem*{assA}{Assumption ${\mathbf H}_p$}
\begin{assA} The restrictions to $\Pc_p$ of the maps $U$ and $F$ satisfy the following.

\medskip
\noindent {\rm (i)} The restriction of the map  $U$ to $\Pc_p$ is continuous;

\medskip
\noindent {\rm (ii)} The map $(x_0,\mu,a,e)\longmapsto F_{x_0}^e(\mu,a)$ is continuous in $(\mu,a)\in\Pc_p(\X \x A)\times A$ for all $(x_0,e)\in\S\x\R$, and has affine growth in $(e,x_0)$ uniformly in $(\mu,a)$.
\end{assA}

Our main result -~reported in the  next subsection~- builds on some given solution of the mean field game in order to construct approximate Nash equilibria. For completeness, we provide two sets of sufficient conditions which guarantee the existence of a solution to the mean field game. 

\medskip
For any $m \in \Pc(\X \x A)$, we denote by $m^X(\mathrm{d}x):=m(\mathrm{d}x,A)$ the marginal of $m$ in $\X$. 

\begin{proposition} \label{prop_continuous}
Let Assumption ${\rm H}_p$ hold for some $p >1$, and let $\eta\in\L^q$ and $\mu_0\in\Pc_q(\S)$ for some $q>p$. Then there exists a solution of the MFG $\muh \in \Pc_p (\X \x A)$ under either one of the following conditions:

\vspace{1mm}
\noindent {\rm (MFG1)} There exists an optimal response map $\widehat{a}:\Pc_p(\X \x A)\longrightarrow \Ac$ satisfying 
	\begin{itemize}
	\item[{\rm (i)}] $J(\mu,\widehat{a}(\mu))=\max_{\alpha\in\Ac} J(\mu,\alpha)$ for all $\mu\in\Pc_p(\S \x A)$,
	\item[{\rm (ii)}] $\mu\longmapsto\widehat{a}(\mu,\bar X_0)$ is a continuous map from $\Pc_p(\X \x A)\longrightarrow A$, $\mu_0 \otimes {\rm U}_{[0,1]}-$a.s.
	\end{itemize}

\vspace{1mm}
\noindent {\rm (MFG2)} Or, $F^e_{x_0}(\mu,a)=\widehat{F}^e_{x_0}\left(\mu^X, \langle f^1, \mu \rangle,\cdots, \langle f^L, \mu \rangle,\,a \right)$  for some Borel maps $\widehat F$ and $(f^\ell)_{1 \le \ell \le L}$, $A$ is convex and the level set 
\begin{equation}\label{cond:conv}
\begin{array}{c}
\big\{ \big(\Lc_{F^\eta_{x_0}\!(\mu,a)},\Fc(x_0,\mu,a),u\big)\!\in\!\Pc(\S)\!\times\!\R^{L+1}\!:
            ~a \!\in\! \Ac~\mbox{independent of $\eta$, and}~u \!\le\! U(\Lc_{F^\eta_{x_0}\!(\mu,a)}) 
\big\}
\\ \\
\mbox{is convex for all}~(x_0,\mu)\in\S\times\Pc(\X\times A).
\end{array}
\end{equation}
where $\Fc(x_0,\mu,a):=( \int_{[0,1]} f^\ell (x_0, F^e_{x_0}\!(\mu,a), a) \,\mathrm{d}e)_{1 \le \ell \le L}$.
\end{proposition}

The proof of this result is deferred to Section \ref{sect:MFGStatic}. We conclude this subsection by providing a sufficient condition for the continuity condition on the optimal response map in Proposition \ref{prop_continuous} (i).

\begin{proposition} \label{prop_CSahat}
In the context of Proposition \ref{prop_continuous} under Condition {\rm (MFG1)}, the continuity of $\mu\longmapsto\widehat{a}(\mu,\bar X_0)$ holds under the following conditions: 
\begin{itemize}
\vspace{-2mm}\item the response map $\widehat{a}$ is unique, and
\vspace{-2mm}\item  for any converging sequence $(\mu_n)_{n\ge 1}$, $\{ \Lc \left(\widehat{a}(\mu_n,x_0, \xi) \right)\}_{n\ge 1}$ is relatively compact in $\Pc_p(A)$ for $\mu_0-$a.e. $x_0\in\S$.
\end{itemize}
\end{proposition}

The proof of this result is also reported in Section \ref{sect:MFGStatic}. 

\subsection{Approximate Nash Equilibrium: main results}

Let $\widehat\mu \in \Pc(\X \x A)$ be a solution of the MFG associated to $\widehat{\alpha} \in \Ac$. We introduce the finite population control strategy:
\begin{eqnarray}\label{alphan}
\boldsymbol{\alpha}^n=(\alpha^{n,i})_{i\le n}\in(\Ac_n)^n
&\mbox{with}&
\alpha^{n.i}(x_0,v) := \widehat{\alpha}(x^i_0,v),
~i=1,\ldots,n,
\end{eqnarray}
Recall the population maximum deviation gain $\overline{\Dbf}_n^\delta[\boldsymbol{\alpha}^n]$, introduced in \eqref{Dn}, and let us also introduce in the present context the population maximum gap to the MFE 
\begin{eqnarray}\label{Gn}
\overline{\Gbf}^\delta_n[x_0,\boldsymbol{\alpha}^n]
&:=&
\max_{i\le n} \big|J_i^n(x_0,\boldsymbol{\alpha}^n)
                            - V(x_0,\widehat\mu) 
                     \big|\1_{\{x_0\in B(n^\delta)\}}.
\end{eqnarray} 
Our main results require the following conditions which involve two parameters $p\ge 0$ and  $r>1$.

\newtheorem*{assB}{Assumption $\mathbf{I}_{r,p}$}
\begin{assB}
{\rm (i)} The sequence of empirical measures of $\{\Theta_0^i:=(X^i_0,\hat\alpha(X^i_0, \xi^i))\}_{i\ge 1}$ converges to the law of $\Theta_0:=(X_0, \hat \alpha (X_0, \xi))$ in $\Wc_r$, i.e. 
$$
\mu^n_{\boldsymbol{\Theta}_0} 
:=
\frac1n\sum_{i=1}^n \delta_{\Theta^i_0}
\;\underset{n\to\infty}{\longrightarrow} \;
\Lc_{\Theta_0}~~\mbox{in }\Wc_r;
$$
{\rm (ii)} There exists a continuous map $c:\S\times A\times\R\longrightarrow[0,1)$ such that 
\begin{eqnarray*}
&\displaystyle
   \sup_{m \neq m'} \frac{|F^e_{x_0}(m,a)- F^e_{x_0}(m',a)|}{ \Wc_p(m,m') }\le c(x_0,a,e),
~\mbox{for all}~
(x_0,a,e)\in \S\times A\times\R;
&
\\
&\mbox{and}~~\E \big[ \big(1- \hat c(\bar X_0, \eta)^{1 \vee p}\big)^{-1} \big] < \infty
~\mbox{with}~\hat c(\bar x_0,e):=c(x_0,\hat\alpha(\bar x_0),e).
&
\end{eqnarray*}
\end{assB}

Notice that the last Condition (i) is obviously satisfied in the special case of an iid sequence $(X^i_0)_{i\ge 1}$. This weaker requirement allows to address more general examples, beyond the iid setting, as in the context of non--exchangeable systems.

\begin{theorem} \label{thm:approx_nash_normal}
Let Assumptions ${\rm H}_p$ and ${\rm I}_{r,p}$ hold for some $r>1$ and $p \in \{0\}\cup[1,2 \wedge r]$, and let $\eta\in\L^r$ and consider a solution of the MFG $\hat\mu \in \Pc_r(\X \x A)$ associated to $\widehat{\alpha} \in \Ac$. Then, the strategy $\boldsymbol{\alpha}^{n}$ defined in \eqref{alphan} is an approximate Nash equilibrium in the following senses:

\medskip\medskip
\noindent \hspace{3mm}{\rm (ANE$_{\L^{r}}$)} \hspace{1mm} $
\overline{\Dbf}^\infty_n[\boldsymbol{\alpha}^n]
\underset{n\to\infty}{\longrightarrow} 0$ and $\overline{\Gbf}^\infty_n[X_0,\boldsymbol{\alpha}^n]
\underset{n\to\infty}{\longrightarrow} 0~\mbox{in}~\L^r(\mu_0).
$

\medskip\medskip
\noindent \hspace{3mm}{\rm (ANE$\,^\delta_{\L^{\infty}}$)} \hspace{1mm} For all $\delta\in(0,\frac1p)$, with $\delta=\infty$ when $p=0$, we have 
$$
\overline{\Dbf}^\delta_n[\boldsymbol{\alpha}^n]
\underset{n\to\infty}{\longrightarrow} 0~\mbox{in}~\L^\infty(\mu_0)
~\mbox{and}~
\overline{\Gbf}^\delta_n[x_0,\boldsymbol{\alpha}^n]
\underset{n\to\infty}{\longrightarrow} 0~\mbox{in}~\L^\infty(\S).
$$
In particular, for any small $\eps>0$, the strategy $\boldsymbol{\alpha}^n$ is an $(\varepsilon,\delta)-$Nash equilibrium in the sense of Definition \ref{def:Nash} for sufficiently  large population size $n$.
\end{theorem}

\begin{remark}\label{rem:interpret_thm_approx_Nash}
{\rm Theorem~\ref{thm:approx_nash_normal} establishes two complementary notions of approximate Nash equilibrium for the strategy profile $\boldsymbol{\alpha}^n$ induced by the mean field equilibrium $\widehat{\alpha}$. The first statement, {\rm (ANE$_{\L^r}$)}, yields convergence of the population maximum deviation gain in $\L^r(\mu_0)$, showing that deviation incentives vanish on average with respect to the distribution of initial states. In particular, this implies that profitable unilateral deviations become asymptotically negligible for a typical agent as the population size grows.

The second statement, {\rm (ANE$^\delta_{\L^\infty}$)}, is stronger, as it provides a uniform control of deviation gains for all agents whose initial states belong to the expanding ball $B(n^\delta)$. This control is also uniform with respect to the agent index \(i\), since it bounds the supremum of deviation gains simultaneously over the entire truncated population. The restriction $\delta<1/p$ (with $\delta=\infty$ when $p=0$) reflects the trade-off between uniform estimates on compact subsets of the state space and the regularity assumptions imposed on the data, in particular on the functional $F$.

Taken together, these results show that the only potential obstruction to a full $\varepsilon$--Nash equilibrium arises from agents whose initial states lie in the far tails of the initial distribution. As the radius $n^\delta$ diverges with $n$, the truncated region captures an increasingly large portion of the population starting points. Consequently, for large population sizes, the strategy $\boldsymbol{\alpha}^n$ satisfies the equilibrium condition of {\rm Definition~\ref{def:Nash} } up to an arbitrarily small error $\varepsilon$ on this dominant subpopulation. This uniform approximation result refines existing asymptotic Nash equilibrium notions and provides a sharper connection between finite-player symmetric games and their mean field game limits.
}
\end{remark}

\begin{proof}[Proof of \Cref{thm:approx_nash_normal}]
We first prove that ANE$_{\L^{r}}$ is a consequence of ANE$^\delta_{\L^{\infty}}$. To see this, we write
\begin{align*}
    \E\left[\overline{\Gbf}^\infty_n[X_0,\boldsymbol{\alpha}^n]^{r}
    \right] 
    &\le 
    \E\left[\overline{\Gbf}^\infty_n[X_0,\boldsymbol{\alpha}^n]^{r}
        \mathbf{1}_{\{ |X_0| \ge n^\delta \}}
    \right]
    + 
    \sup_{x_0\in B(n^\delta)}\overline{\Gbf}^\infty_n[x_0,\boldsymbol{\alpha}^n]^{r}
    \\
    &\le 
    \E\left[\overline{\Gbf}^\infty_n[X_0,\boldsymbol{\alpha}^n]^{2r}
    \right]^{1/2}
    \P\big[ |X_0| \ge n^\delta \big]^{1/2}
    + 
    \sup_{x_0\in B(n^\delta)}\overline{\Gbf}^\infty_n[x_0,\boldsymbol{\alpha}^n]^{r},
\end{align*}
where the last term on the right-hand side converges to zero by ANE$_{\L^{\infty}}$. We also have that $\P[ |X_0| \ge n^\delta]\longrightarrow 0$ as $n\to\infty$ and 
$$
       \E\Big[
        \sup_{n \ge 1}\overline{\Gbf}^\infty_n[X_0,\boldsymbol{\alpha}^n]^{2r}
          \Big]
        =\lim_{K\to\infty}\!\!\!\!\uparrow
       \E\Big[
        \sup_{n \ge 1}\overline{\Gbf}^\infty_n[X_0,\boldsymbol{\alpha}^n]^{2r}
        \mathbf{1}_{\{ |X_0| \le K \}}
          \Big] 
          \le 
          \sup_{n \ge 1}\sup_{|x_0|\le n^\delta}
          \!\!\overline{\Gbf}_n[x_0,\boldsymbol{\alpha}^n]^{2r}
    < \infty,
    $$
again by ANE$^\delta_{\L^{\infty}}$. We deduce the required convergence $\E\big[\overline{\Gbf}^\infty_n[X_0,\boldsymbol{\alpha}^n]^{r}\big] \longrightarrow 0$ as $n\to\infty$. The same argument applies to show that $\E\big[\overline{\Dbf}^\infty_n[\boldsymbol{\alpha}^n]^r\big]
\underset{n\to\infty}{\longrightarrow} 0$.

\vspace{2mm}

\noindent We next prove ANE$^\delta_{\L^{\infty}}$ in several steps.

\medskip
\noindent {\bf 1.} {\it Propogation of chaos.} Denote $\mathbf{X}^n:=\mathbf{X}^{n,\boldsymbol{\alpha}^n}$ with $\boldsymbol{\alpha}^n$ defined by \eqref{alphan} through the MFE strategy $\widehat\alpha$, and define $\widehat{\Xbb}^n_1:=(\widehat{X}^{n,1}_1,\dots, \widehat{X}^{n,n}_1)$ by substituting the MFE $\widehat\mu$ to the empirical measure:
\begin{eqnarray}\label{Xhat}
    \widehat{X}^{n,i}_1=F^{\eta^i}_{X^i_0} (\hat\mu, \alpha^{n,i} ),~i=1,\ldots,n,
\end{eqnarray}
and denote $\boldsymbol{\hat\Theta}^n=(\widehat{\Theta}^{i,n})_{1\le i\le n}$ with $\widehat{\Theta}^{i,n}:=(\widehat{X}^{i,n},\alpha^{n,i})$. Then, with $\{\Theta^i_0:=(X^i_0,\hat\alpha(\bar X^i_0))\}_{i\ge 1}$ as in Assumption ${\rm I}_{r,p}$ (i) and denoting $\mub^n_{\boldsymbol{\Theta}^n_0}:=\frac{1}{n} \sum_{i=1}^n \delta_{(\Theta_0^i,\eta^i)}$, we have 
\begin{align}
    \mu^n_{\boldsymbol{\hat\Theta}^n}
    &:= 
    \frac{1}{n} \sum_{i=1}^n \delta_{\widehat{\Theta}^{i,n}}
    =
    \int_{\S \x A \x \R} \delta_{(x_0,F^{e}_{x_0} (\hat\mu, a ), a)} \mub^n_{\boldsymbol{\Theta}^n_0}(\mathrm{d}x_0,\mathrm{d}a,\mathrm{d}e)
    \stackrel{\Wc_r}{-\!\!-\!\!\!\longrightarrow}
    \Lc_{(X^{\hat\mu,\hat\alpha},\, \hat \alpha (\bar X_0))}=\muh
    \label{eq:strong_poc}
\end{align}
as $n \to \infty$, by Assumption ${\rm I}_{r,p}$ (i), as $(\bar X^i_0)$ is independent of $\eta_i$, the sequence $(\eta_i)_{i\ge 1}$ is iid, and $F$ is continuous and has linear growth.

Let $\widehat c(\bar x_0,e):=c(x_0,\hat\alpha(\bar x_0),e)$ as in Assumption ${\rm I}_{r,p}$ and denote $\boldsymbol{\Theta}^n:=(\Xbb^n,\alphab^n)$. By the Lipschitz property of $F$ in Assumption ${\rm I}_{r,p}$ (ii) and the triangle inequality, we see that
\begin{eqnarray}
    \left| \widehat{X}^{n,i}_1 - {X}^{n,i}_1 \right| 
    &=& 
    \big| F^{\eta^i}_{X^i_0} (\hat\mu, \hat{\alpha} (\bar X^{i}_0,\hat\mu))
            - F^{\eta^i}_{X^i_0} (\mu^n_{\boldsymbol{\Theta}^n}, \hat{\alpha} (\bar X^{i}_0, \hat\mu))
    \big|
    \nonumber\\
    &\le&
    \hat c(\bar X^i_0\eta^i)\big[ \Wc_p (\mu^n_{\boldsymbol{\widehat\Theta}^n}, \hat\mu) 
                                  + \Wc_p (\mu^n_{\boldsymbol{\widehat\Theta}^n}
                                                  , \mu^n_{\boldsymbol{\Theta}^n}) 
                        \big].
\label{XniXbarni}
\end{eqnarray}
Then, denoting $q:=1 \vee p$ and recalling that $p<2$, it follows that
\begin{align*}
\Wc_p \big(\mu^n_{\boldsymbol{\widehat\Theta}^n}, \mu^n_{\boldsymbol{\Theta}^n}\!
           \big)^q
&\le 
\frac{1}{n} \sum_{i=1}^n \left|X^{n,i}_1\!-\! \widehat{X}^{n,i}_1 \right|^{q}
\le
\big[ \Wc_p (\mu^n_{\boldsymbol{\widehat\Theta}^n}, \hat\mu)^q 
                                  \!+\! \Wc_p (\mu^n_{\boldsymbol{\widehat\Theta}^n}
                                                  , \mu^n_{\boldsymbol{\Theta}^n})^q 
                        \big]
\frac{1}{n} \sum_{i=1}^n \hat c(X^i_0,\xi^i,\eta^i)^q,
\end{align*}
where the map $\hat c$ takes values in $[0,1)$ by Assumption ${\rm I}_{r,p}$ (ii). This provides 
\begin{equation}\label{XnXbarn}
    \Wc_p \big(\mu^n_{\boldsymbol{\widehat\Theta}^n}, \mu^n_{\boldsymbol{\Theta}^n}\!
           \big)^q
    \le C_n \Wc_p \big(\mu^n_{\boldsymbol{\widehat\Theta}^n}, \hat\mu\big)^q,
    ~\mbox{a.s. with}~
    C_n
    :=
    \frac{\frac{1}{n} \sum_{i=1}^n \hat c(X^i_0,\xi^i,\eta^i)^q}
           {1-\frac{1}{n} \sum_{i=1}^n \hat c(X^i_0,\xi^i,\eta^i)^q} 
            .
\end{equation}
As $X^i_0$, $\xi^i$ and $\eta^i$ are independent, and   $(\eta_i)_{i\ge 1}$ is iid, it follows from Assumption ${\rm I}_{r,p}$ (i) that $C_n\longrightarrow \frac{\E[\hat c(X_0,\xi,\eta)^q]}{1-\E[\hat c(X_0,\xi,\eta)^q]}<\infty$, a.s. as $n\to\infty$. Then, we see from \eqref{XnXbarn} together with \eqref{eq:strong_poc} that $\Wc_p(\mu^n_{\boldsymbol{\widehat\Theta}^n}, \mu^n_{\boldsymbol{\Theta}^n})\longrightarrow 0$, a.s. and therefore
$$
\Wc_p \big(\mu^n_{\boldsymbol{\Theta}^n},\hat\mu\big)\longrightarrow 0,
~\mbox{and}~ 
 \sup_{i \le n} \left| \widehat{X}^{n,i}_1 \!-\! {X}^{n,i}_1 \right|^q 
    \stackrel{n \to \infty}\longrightarrow 0,
~\mbox{a.s. for all}~i \le n,
$$
by \eqref{XniXbarni} together with the fact that $\widehat{c} \le 1$. 

\medskip
\noindent {\bf 2.} In this step, we fix some $\delta<\frac1{p}$ and show that
\begin{align} \label{eq:uniform_state_index}
    H^n:=
    \sup_{i \le n} \sup_{x_0\in B(n^\delta)}\E \left[\left| \widehat{X}^{n,i}_1 \!-\! {X}^{n,i}_1 \right|^q \mid X^i_0=x_0 \right]
    \underset{n\to\infty}{\longrightarrow} 0~\mbox{as}~n\to\infty,
\end{align}
which in turn implies that
$$
\sup_{i \le n} \sup_{x_0\in B(n^\delta)}
\Wc_p \left( \Lc^{|X^i_0=x_0}_{X^{n,i}_1},\, \Lc^{|X^i_0=x_0}_{\Xh^{n,i}_1} \right)^q 
\le 
H^n \stackrel{n \to \infty}\longrightarrow 0,
$$
and we may then deduce our first required result by the continuity of $U$ in Assumption ${\rm H}_p$:
\begin{align*}
\overline{\Gbf}_n[x_0,\boldsymbol{\alpha}^n]\1_{\{x_0\in B(n^\delta)\}}
\underset{n\to\infty}{\longrightarrow} 0~\mbox{in}~\L^\infty(\S).
\end{align*}
To prove \eqref{eq:uniform_state_index}, let $(i_n)_{n \ge 1}$ and $(x^n_0)_{n \ge 1}$ be the sequence of maximizers of $H^n$.
Consider again the state process $\widehat\Xbb^n$ defined in \eqref{Xhat} and gthe corresponding $\boldsymbol{\Theta}^n$ with initial state $X^{i_n}_0=x^{n}_0\in B(n^\delta)$, or equivalently $|x^{n}_0|^p \le n^{1-\varepsilon}$ with $\varepsilon:=1-\delta p>0$. Then, we estimate for all continuous map $\varphi$ with $p-$polynomial growth that:
\begin{align}\label{crucial}
\langle\varphi,\mu^n_{\boldsymbol{\widehat\Theta}^n}\rangle
:=
\int \!\varphi \,d\mu^n_{\boldsymbol{\widehat\Theta}^n}
&=
\frac{1}{n} \sum_{j=1}^n \varphi(X^j_0,\widehat X^j_1,\hat \alpha (\bar X^j_0)) \nonumber
\\
&
=
\frac1n\varphi\big(x_0^{n} ,
                             F^{\eta^{i_n}}_{x_0^{n}}\big(\hat\mu,
                                                                     \hat\alpha(x_0^{n},\xi^{i_n})
                                                              \big), \hat \alpha (x^n_0,\xi^{i_n})
                      \big) 
+ \phi_n,
\end{align}
with $\phi_n:=\frac{1}{n} \sum_{j=1,\, j \neq i_n }^n 
\varphi(X^j_0,\widehat X^j_1, \hat \alpha (\bar X^j_0))\longrightarrow \langle \varphi,\hat\mu\rangle$, a.s. by Assumption ${\rm I}_{r,p}$ (i) due to the fact that $\widehat X^j_1=F^{\eta^j}_{X^j_0}(\hat\mu, \hat\alpha(\bar X^j_0))$ together with the independence between $\eta^j$, $\xi^j$ and $X^j_0$. In addition, the nonlinearity $F$ has affine growth in $(x_0,e)$ uniformly in $(\mu,a)$ by Assumption ${\rm H}_p$ (ii). Then it follows from the boundedness condition on $\rho_A$ together with the $p-$polynomial growth of $\varphi$ that
\begin{align*}
\frac{1}{n} \varphi\big( x^{n}_0, F^{\eta^{i_n}}_{x^{n}_0} \big( \muh, \widehat{\alpha}(x^{n}_0,\xi^{i_n}) \big),\,\alphah (x^n_0,\xi^{i_n}) \big) 
\le 
\frac{C}{n}\big(|x^{n}_0|^p + |\eta^{i_n}|^p  \big) 
\le 
\frac{C}{n}\big(n^{1-\varepsilon} + |\eta^{i_n}|^p  \big) 
\underset{n \to \infty}{\longrightarrow} 0,
\end{align*}
regardless of the behavior of the sequence $(i_n,x_0^{n})_n$. By \eqref{crucial} this yields $\langle\varphi,\mu^n_{(\widehat\Xbb^n,\alphab^n)}\rangle\longrightarrow \langle \varphi,\hat\mu\rangle$, a.s. and it follows from the arbitrariness of $\varphi$ that
\begin{align} \label{eq:unif_empirical}
    \lim_{n \to \infty} 
\E \big[\Wc_p(\mu^n_{\hat{\boldsymbol{\theta}}^n},\hat\mu)^q \big| X^{i_n}_0\!=\!x_0^{n} \big]=0,
~\mbox{and}~
\lim_{n \to \infty} 
\E \big[\Wc_p(\mu^n_{\hat{\boldsymbol{\theta}}^n},\hat\mu)^q \big| X^{i_n}_0\!=\!x_0^{n} \big]=0,
\end{align}
which implies \eqref{eq:uniform_state_index} by using \eqref{XniXbarni}.

\medskip

\noindent {\bf 3. {\it Optimization and deviating player.}} 
For all $n \ge 1$ and $i=1,\ldots, n$, we denote by $\betab^{i,n}:=\alpha^{n,(-i)}\oplus\beta$ and $\Ybb^{n,i}:=\Xbb^{n,\betab^{i,n}}$, where we recall that $\alpha^{n,(-i)}\oplus\beta$ is obtained from $\boldsymbol{\alpha}^n$ by replacing the $i-$th coordinate $\alpha^{n,i}$ of $\boldsymbol{\alpha}^n$ by the arbitrary deviating strategy $\beta \in \Ac_n$. We also denote $\boldsymbol{\Theta}^{n,i,\beta}=(\Ybb^n,\boldsymbol{\beta}^{i,n})$ and we recall that $\boldsymbol{\Theta}^n=(\Xbb^n,\boldsymbol{\alpha})$. With $q=1 \vee p$ and using the decomposition $\mu^n_{\boldsymbol{\Theta}^{n,i,\beta}}=\frac{1}{n} (\delta_{(X^i_0,\Ybb^{n,i,i}_1, \beta(\bar X_0))}+\sum_{j\neq i} \delta_{(X^j_0,\Ybb^{n,i,j}_1, \hat \alpha(X^j_0,\xi^j))})$, we obtain by following the same line of argument as in the previous step:
\begin{align*}
    \Wc_p( \mu^n_{\boldsymbol{\Theta}^n},\mu^n_{\boldsymbol{\Theta}^{n,i,\beta}})^{q} 
    &\le 
    \frac{1}{n}\|\rho_A\|^q_{\L^\infty_{A\times A}}
    +|Y^{n,i,i}_1 - {X}^{n,i}_1 |^{q}
    +\sum_{j \neq i} |Y^{n,i,j}_1 - {X}^{n,j}_1 |^{q}
    \\
    &\le \frac{1}{n}\|\rho_A\|^q_{\L^\infty_{A\times A}} 
       + |Y^{n,i,i}_1 - {X}^{n,i}_1 |^{q}
    +\sum_{j \neq i} \hat c (\bar X^j_0,\eta^j)^q\,
                                                 \Wc_p ( \mu^n_{\boldsymbol{\Theta}^n}, \mu^n_{\boldsymbol{\Theta}^{n,i,\beta}})^{q} .
\end{align*}
Denoting $\tilde C_{n,i}^{-1} := 1 - \frac{1}{n} \sum_{j \neq i} \hat c (X^j_0,\xi^j,\eta^j)^q$, this leads to:
\begin{align*}
    \Wc_p( \mu^n_{\boldsymbol{\Theta}^n},\mu^n_{\boldsymbol{\Theta}^{n,i,\beta}})^{q} 
    \le 
    \tilde C_{n,i}\;
    \frac{1}{n} \big( \|\rho_A\|^q_{\L^\infty_{A\times A}} + \big|Y^{n,i,i}_1 - {X}^{n,i}_1 \big|^{q} \big).
\end{align*}
Then, for any sequence $(i_n,x^n_0)_{n \ge 1}$ with $i_n \in \{1,\cdots, n\}$ and $|x_0^n|^p \le n^{1- \varepsilon}$, we obtain
\begin{align*}
    \E\big[\Wc_p( \mu^n_{\boldsymbol{\Theta}^n},\mu^n_{\boldsymbol{\Theta}^{n,i_n,\beta}})^{q} 
               | X^{i_n}_0=x^n_0 
        \big] 
    \le 
    \frac{C}{n} (1+|x^n_0|^q)
    \E[ \tilde C_{n,i_n} | X^{i_n}_0\!=\! x^n_0 ] ,
\end{align*}
for some constant $C >0$ independent of $n$, depending only on $\|\rho_A\|^q_{\L^\infty_{A\times A}}$ and the linear growth parameter of $F$. By Assumption ${\rm I}_{r,p}$ (ii), it follows that $\lim\E[ \tilde C_{n,i_n} | X^{i_n}_0\!=\! x^n_0 ]= \E \big[ (1\!-\!c (X_0, \hat \alpha (X_0,\xi), \eta)^{1 \vee p})^{-1}\big] <\infty$, this shows that $\E\big[\Wc_p( \mu^n_{\boldsymbol{\Theta}^n},\mu^n_{\boldsymbol{\Theta}^{n,i_n,\beta}})^{q} 
               | X^{i_n}_0\!=\! x^n_0 
        \big]\longrightarrow 0$ as $n\to\infty$, and by the arbitrariness of $(i_n,x^n_0)_{n \ge 1}$, we deduce that
$$
    \Lim_{n \to \infty}\sup_{i \le n} \sup_{x_0\in B(n^{\delta})}
    \E\left[\Wc_p( \mu^n_{\boldsymbol{\Theta}^n},\mu^n_{\boldsymbol{\Theta}^{n,i_n,\beta}})^{q} | X^{i_n}_0=x_0 \right] =0.
$$
Combining with \eqref{eq:unif_empirical}, we see that: 
$$
\lim_{n \to \infty} \sup_{i \le n} \sup_{x_0\in B(n^{\delta})} 
\E \big[\Wc_p(\mu^n_{\boldsymbol{\Theta}^{n,i_n,\beta}},\hat\mu)^q \big| X^i_0=x_0 \big]=0
$$
and
$$
    \lim_{n \to \infty}  \sup_{i \le n} \sup_{x_0\in B(n^{\delta})} \sup_{j \neq i} \E \big[ |X^{n,j}_1-Y^{n,i,j}_1|^q 
                                   \big| X^i_0=x_0 
                                   \big]=0.
$$

\medskip 
\noindent {\bf 4.} {\it Approximate NE.} 
Denote $\gamma^n:=\sup_{i \le n}\| D_i^n[\boldsymbol{\alpha}^n] \mathbf{1}_{\{X^i_0\in B(n^\delta)\}} \|_{\L^\infty(\mu_0)}$, and let $(i_n)_{n \ge 1} \subset \N^\star$ and $(\beta^{i,n},x_0^{i,n})_{1 \le i \le n}$ be such that
\begin{align}\label{deltain}
\gamma^n - 2^{-n} &\le J_{i_n}^n(x_0^{i_n,n},\alpha^{n,(-i_n)}\!\oplus\!\beta^{i_n,n})
    -
    J_{i_n}^n(x_0^{i_n,n},\alpha^n). \nonumber 
\end{align}
Applying the result of Step 2 to $\Ybb^{n}:=\Xbb^{n,\alpha^{n,(-i_n)}\oplus\beta^{i_n,n}}$, we deduce that
\begin{align*}
    \Wc_p \Big( \Lc^{| X^{i_n}_0=x_0^{i_n,n}}_{Y^{n,i_n}_1},
                        \Lc_{F^{\eta^{i_n}}_{X_0^{i_n}}(\hat\mu,\beta^{i_n,n}(\bar\Xbb^n_0))}
                             ^{| X^{i_n}_0=x_0^{i_n,n}}\Big) 
    \le 
    \E \big[ \Wc_p \left( \mu^n_{\boldsymbol{\Theta}^{n,i_n,\beta}},\muh \right)  
                \big| X^{i_n}_0 = x^{i_n,n}_0
         \big] 
    \stackrel{n \to \infty}\longrightarrow 0.
\end{align*}
By our continuity conditions on $U$ in Assumption ${\rm H}_p$, this provides
\begin{align*}
    \lim_{n \to \infty} \Big|  J_{i_n}^n(x^{i_n,n}_0,\alpha^{n,(-i_n)}\!\oplus\!\beta^{i_n,n})
                                        - U \Big(\Lc_{F^{\eta^{i_n}}_{X_0^{i_n}}(\hat\mu, \beta^{i_n,n}(\Xbb^n_0,\boldsymbol{\xi}^n))}^{| X^{i_n}_0=x_0^{i_n,n}}\Big) 
                               \Big|=0,
\end{align*}
and
\begin{align*}
    \lim_{n \to \infty}\Big| J_{i_n}^n(x^{i_n,n}_0,\alpha^n) 
                                     - U \big( \Lc_{X^{\mu,\hat{\alpha}}_1}
                                                        ^{| X_0=x_0^{i_n,n}}
                                           \big) 
    \Big|=0.
\end{align*}
Then it follows from the definition of $\delta^n$ that:
\begin{eqnarray}
0\;\le\;
\limsup_{n \to \infty} \gamma^n 
&\le&
\limsup_{n \to \infty} J_{i_n}^n(x^{i_n,n},\alpha^{n,(-i_n)}\!\oplus\!\beta^{i_n,n})
-
J_{i_n}^n(x^{i_n,n},\alpha^n)
\nonumber \\
&=&
\limsup_{n \to \infty} U \Big(\Lc_{F^{\eta^{i_n}}_{x_0^{i_n,n}}
                                                     (\hat\mu, \beta^{i_n,n}(\Xbb^n_0,\boldsymbol{\xi}^n))}
                                                 ^{| X^{i_n}_0=x_0^{i_n,n}} 
                                     \Big)
                                 -U \big(\Lc_{X^{\hat\mu,\hat{\alpha}}_1}
                                                  ^{| X_0=x_0^{i_n,n}}
                                      \big).
\label{conv-relax}
\end{eqnarray}
By \citeauthor*{blackwellDubins83} \cite{blackwellDubins83}, we may find a Borel map $M: [0,1]\x \S \to \S^{n-1} \x [0,1]^n$ such that 
$$
\Lc_{M (\xi,x)}= \Lc_{(\Xbb^n_0, \boldsymbol{\xi}^n)}^{| X^i_0=x}, 
~\mbox{and therefore}~
\Lc_{(\eta^i,\beta^{i,n}(\Xbb^n_0,\boldsymbol{\xi}^n))}^{|X^i_0=x^{i,n}_0}=\Lc_{(\eta^i, \beta^{i,n} (M(\xi,x^{i,n}_0)) )}.
$$ 
As $\bar x_0 \mapsto \beta^{i,n} (M(\bar x_0)) $ belongs to $\Ac$, it follows that
\begin{align*}
U \Big(\Lc_{F^{\eta^{i_n}}_{x_0^{i,n}}
                    (\hat\mu, \beta^{{i_n},n}(\Xbb^n_0,\boldsymbol{\xi}^n))}
                ^{| X^{i_n}_0=x_0^{{i_n},n}} 
   \Big)
&\le \sup_{\beta^n \in \Ac_n}  
       U \Big( \Lc_{F^{\eta^{i_n}}_{X_0^{i_n}}(\hat\mu, \beta^{n}(\Xbb^n_0,\boldsymbol{\xi}^n))}
                        ^{| X^{i_n}_0=x_0^{{i_n},n}} 
       \Big)
\\
&= \sup_{\alpha \in \Ac} U \big( \Lc_{X^{\hat\mu,\alpha}_1}
                                                     ^{|X_0=x_0^{{i_n},n}} 
                                        \big) 
= U \big(\Lc_{X^{\hat\mu,\hat{\alpha}}_1}
                  ^{| X_0=x_0^{{i_n},n}}
       \big),
\end{align*}
which in view of \eqref{conv-relax} implies that $\gamma^n\longrightarrow 0$ as $n\to\infty$.
\end{proof}

\section{Continuous time approximate Nash equilibria}
\label{sec:continuous}

In this section, we extend our approximate Nash equilibria results to the continuous time setting. For simplicity, we consider a coupling through the law of the state process only, and not the joint law of the state-control as in the the previous one period setting. We shall also restrict for simplicity to the uncontrolled diffusion case.

\subsection{The continuous time mean field game}

Let $p \in \{0\} \cup [1,\infty)$, $W$ be a Brownian motion on the filtered probability space $(\Om,\Fc,\F,\P)$, and
$$
b:[0,T]\times\R^{2d}\times\Pc_p(\R^{2d})\times A
   \longrightarrow \R^d.
$$
Given an initial condition $X_0$ independent of $W$, and a measurable flow of probability measures $(\mu_t)_{t\in[0,T]} \subset \Pc(\R^{2d})$, we define the controlled state process $X^{\mu,\alpha}$ through the controlled SDE
\begin{eqnarray}\label{SDE}
\mathrm{d}X_t 
= b\big(t,\underline{X}_t,\mu_t,\alpha(t,\underline{X}_t)\big)\,\mathrm{d}t 
       + \mathrm{d}W_t, 
~\mbox{with}~\underline{X}_t:=(X_0,X_t),~ t\in[0,T].
\end{eqnarray}
Here, the control strategy $\alpha$ lies in the set $\Ac$ of bounded Borel $\F-$progressively measurable maps from $[0,T] \times \R^{2d}$ to some compact set $A$ of a Polish space. We shall impose appropriate conditions on $b$ to guarantee the existence of a unique solution $X^{\mu,\alpha}$. 

Similar to the static setting, we consider a performance criterion defined through a map $U:\Pc_p(\R^d)\longrightarrow\R$, and we define the optimization problem
 \[
 V(\mu):=\sup_{\alpha\in\Ac}J(\mu,\alpha)
  ~\mbox{where}~
  J(\mu,\alpha) 
  := 
  \E\big[U\big(\Lc^{|X_0}_{X^{\mu,\alpha}_T}\big)
     \big],
\]
By standard measurable selection, it follows that
\[
 V(\mu):=\E[V(X_0,\mu)]
  ~\mbox{with}~
  V(X_0,\mu) 
  = 
  \esup_{\alpha} U\big(\Lc^{|X_0}_{X^{\mu,\alpha}_T}\big),
\]
where the dependence of $\alpha$ on $X_0$ in the last $\esup$ is superfluous. We next introduce a precise definition of a solution to the mean field game by the standard representative agent individual optimality (IO), given the environment, and the consistency property expressed as a fixed point (FP) condition.

\begin{definition}
A measurable flow $\widehat{\mu} = (\widehat\mu_t)_{t \in [0,T]} \subset \Pc(\R^{2d})$ 
is a mean field game (MFG) solution associated with a control 
$\widehat{\alpha} \in \Ac$ if the following two conditions hold:
\begin{enumerate}
    \item[\textnormal{(IO)}] 
    Given $\widehat\mu$, the control $\widehat{\alpha}$ solves the individual objective, i.e. $J(\widehat\mu,\widehat{\alpha}) = \sup_{\alpha \in \Ac} J(\widehat\mu,\alpha);$
\item[\textnormal{(FP)}] The optimally controlled dynamics is consistent with the given envioronment $\widehat\mu$, i.e. $\P \circ (X_0,X^{\hat\mu,\hat{\alpha}}_t)^{-1} = \widehat\mu_t,$ for all $t \in [0,T].$
\end{enumerate}
\end{definition}

The following condition is needed to guarantee the existence of a solution to the MFG problem. 

\newtheorem*{assE1}{Assumption $\mathbf{E}_p$}
\begin{assE1}
{\rm (i)} $U$ is bounded, concave and continuous in $\Pc_p(\R^d)$;

\medskip 
\noindent 
{\rm (i)} $b$ is bounded, Lipschitz in $(x,m)$ uniformly in $(t,a)$, and continuous in $(x,m,a)$ for all fixed $t$;

\medskip 
\noindent {\rm (iii)} The control space $A$ is convex and, $(t,x,m)\in[0,T]\times\R^{2d}\times\Pc(\R^{2d})$ the set
\[
    \big\{b(t,x,m,a)\in\R^{d}:~a\in A
    \big\}
    ~\mbox{is convex.}
\]
\end{assE1}

\begin{theorem}
Let Assumption ${\rm E}_p$ hold and let $\mu_0\in\Pc_q(\R^d)$, for some $q>p$. Then, there exists an {\rm MFG} solution.
\end{theorem}

The proof is reported in Section \ref{sect:MFGStatic}.

\section{Continuous time approximate Nash equilibria}

Let $(W^i)_{i \ge 1}$ be a sequence of independent standard Brownian motions, and let $(X_0^i)_{i \ge 1}$ be an i.i.d.\ sequence of random variables with common distribution $\Lc(X_0^i)=\Lc(X_0)$. For each $n \ge 1$, we denote by $\Ac_n$ the set of $A$--valued processes $(\alpha_t)_{t \in [0,T]}$ that are predictable with respect to the filtration $\F^n=\{\Fc^n_t\}_{t \in [0,T]}$ with
\[
    \Fc^n_t:=\sigma(X_0^i,W^i_{\wedge t}, i=1,\ldots,n),~~ t\in [0,T].
\]

Given a vector of controls $\alphab:=(\alpha^1,\ldots,\alpha^n) \in (\Ac_n)^n$, we consider the $n$--tuple of controlled states $\Xbb^{\alphab}:=(X^{1,\alphab},\ldots,X^{n,\alphab})$, where each component satisfies the dynamics: 
\begin{eqnarray*}
X^{1,\alphab}_0=X^i_0
~\mbox{and}~
    \mathrm{d}X^i_t = b (t,\underline{X}^i_t,\mu^n_t,\alpha^i_t)\mathrm{d}t + \mathrm{d}W^i_t, 
    &\mbox{with}&
    \mu^n_t := \frac{1}{n}\sum_{j=1}^n \delta_{\underline{X}^j_t},
    ~t \in [0,T].
\end{eqnarray*}
The maximal value that player $i$ can achieve by deviating unilaterally is
\begin{eqnarray*}
    V_i^n(X_0^i, \alphab^{-i})
    := \sup_{\alpha^i \in \Ac_n} J_i^n(X_0^i,\alphab),
    &\mbox{with}&
    J_i^n(X_0^i,\alphab)
    := U \Big(\Lc^{|X_0^i}_{X^{i,\alphab}_T} \Big).
\end{eqnarray*}
Similar to \eqref{Dn} in the static setting, we introduce the $\delta-$truncated maximum deviation gain for all $\delta\in(0,\infty]$: 
$$
    \overline{\Dbf}^\delta_n[\boldsymbol{\alpha}]
    :=
    \max_{i\le n}\Drm_n^i[\boldsymbol{\alpha}]\1_{\{X^i_0\in B(n^\delta)\}}
~\mbox{with}~
\Drm_n^i[\boldsymbol{\alpha}]
:=
V_i^n(X^i_0,\boldsymbol{\alpha}^{-i})-J_i^n(X^i_0,\boldsymbol{\alpha}),
~i=1,\ldots,n.
$$

\begin{definition}[Open--loop Nash equilibirum]

For fixed $n\ge 1$ and $\varepsilon>0$, a strategy $\boldsymbol{\alpha}\in(\Ac_n)^n$ is called
\\
{\rm (NE)} a Nash equilibrium for $n-$player game if $\overline{\Dbf}^\infty_n[\boldsymbol{\alpha}]=0$;
\\
{\rm (NE$_\varepsilon$)} an $\varepsilon-$Nash equilibrium for $n-$player game if $\overline{\Dbf}^\infty_n[\boldsymbol{\alpha}]\le\varepsilon$.\\
{\rm (NE$^\delta_\varepsilon$)} an $(\varepsilon,\delta)-$Nash equilibrium for $n-$player game if $\overline{\Dbf}^\delta_n[\boldsymbol{\alpha}]\le\varepsilon$.
\end{definition}

\medskip
Let $\widehat{\mu}$ be a solution of the mean field game (MFG) associated with the closed--loop control $\widehat{\alpha}\in \Ac$. As the drift coefficient $b$ is bounded, by Assumption ${\rm E}_p$ (i), it follows from \citeauthor*{AJVeretennikov_1981} \cite{AJVeretennikov_1981} that $X^{\hat\mu,\hat\alpha}$ is strong solution, i.e.
\begin{eqnarray}
X^{\hat\mu,\hat\alpha}_t=\Phi(t,X_0,W_{\wedge t}),
~~t\le T,
\end{eqnarray} 
for some progressively measurable map $\Phi:[0,T] \x \R^d\x C([0,T];\R^d) \longrightarrow \R^d$. 

For all $n \ge 1$, we define the strategy profile 
\begin{equation}\label{alphan:continuous}
    \alphab^n := (\alpha^{1,n},\ldots,\alpha^{n,n}) \in (\Ac_n)^n, 
    ~\mbox{where}~ \alpha^{i,n}_t := \widehat\alpha(t,X_0^i, \Phi (t, X^i_0,W^i_{t \wedge \cdot} ) ),
\end{equation}
and we consider the corresponding $n$--player system $(X^1,\ldots,X^n)$ started from $(X^i_0)_{1 \le i \le n}$ and defined by:
\begin{equation}\label{SDE:ANE}
    \mathrm{d}X^i_t 
    = b ( t,\underline{X}^i_t,\mu^n_t,\alpha^{i,n}_t )\mathrm{d}t 
    + \mathrm{d}W^i_t, 
    ~\mbox{with}~
    \mu^n_t := \frac{1}{n}\sum_{j=1}^n \delta_{\underline{X}^j_t},~
    t \in [0,T],~i=1,\ldots,n.
\end{equation}
Similar to \eqref{Gn}, we also introduce the gap to the MFE:
\begin{eqnarray*}
\overline{\Gbf}^\delta_n[x_0,\boldsymbol{\alpha}^n]
:=
\max_{i\le n} \big|J_i^n(x_0,\boldsymbol{\alpha}^n)
                            - V(x_0,\hat\mu) 
                     \big|\1_{\{x_0\in B(n^\delta)\}}.
\end{eqnarray*}

\begin{theorem}
Let Assumption ${\rm E}_p$ hold and let $\mu_0\in\Pc_q(\R^d)$, for some $q>p$. Then,  the strategy $\boldsymbol{\alpha}^{n}$ defined in \eqref{alphan:continuous} is an approximate Nash equilibrium in the following senses:

\medskip\medskip
\noindent \hspace{3mm}{{\rm (ANE$_{\L^{r}}$)}} \hspace{1mm} For all $r>1$, we have 
$\overline{\Dbf}^\infty_n[\boldsymbol{\alpha}^n]
\underset{n\to\infty}{\longrightarrow} 0
~\mbox{and}~\overline{\Gbf}^\infty_n[X_0,\boldsymbol{\alpha}^n]
\underset{n\to\infty}{\longrightarrow} 0~\mbox{in}~\L^r(\mu_0).
$

\medskip\medskip
\noindent \hspace{3mm}{{\rm (ANE$^\delta_{\L^\infty}$)}} \hspace{1mm} For all $\delta\in(0,\frac1p)$, with $\delta=\infty$ when $p=0$, we have 
$$
\overline{\Dbf}^\delta_n[\boldsymbol{\alpha}^n]
\underset{n\to\infty}{\longrightarrow} 0~\mbox{in}~\L^\infty(\mu_0)
~\mbox{and}~
\overline{\Gbf}^\delta_n[x_0,\boldsymbol{\alpha}^n]
\underset{n\to\infty}{\longrightarrow} 0~\mbox{in}~\L^\infty(\S).
$$
In particular, for all $\eps>0$, the strategy $\boldsymbol{\alpha}^n$ is an $(\eps,\delta)-$Nash equilibrium for sufficiently large population size $n$.
\end{theorem}

\begin{proof}
First, {\rm (ANE$_{\L^{r}}$)} is a direct consequence  of {\rm (ANE$^\delta_{\L^{\infty}}$)}. We omit this part of the proof as it follows by following the same arguments as in the static setting, and we now focus on the proof of {\rm (ANE$^\delta_{\L^{\infty}}$)}.

\medskip

\noindent\textit{Step 1: Convergence of Nash equilibria.}  
Let $(\Xb^i)_{i\ge1}$ be the processes defined by the initial conditions $\Xb^i_0 = X^i_0$ and by replacing the empirical measure $\mu^n$ in the controlled SDE \eqref{SDE:ANE} by the deterministic MFG measure flow $\widehat\mu$:
\begin{align*}
    \mathrm{d}\Xb^i_t
    &= b\bigl(t,\Xb^i_t,\widehat\mu_t,\alpha^{i,n}_t \bigr)\mathrm{d}t 
       + \mathrm{d}W^i_t, 
       \qquad t\in[0,T].
\end{align*}
By construction, $(\Xb^i)_{i\ge1}$ forms a sequence of i.i.d.\ copies of the mean--field state process $X^{\hat\mu,\hat\alpha}$. Then the corresponding sequence of empirical measures $\mub^n := \frac1n \sum_{i=1}^n \delta_{\Xb^i}$ converge by the law of large numbers, together with uniform integrability:
\begin{align} \label{eq:poc}
    \lim_{n\to\infty}
    \Wc_p\bigl(\mub^n,\,\Lc_{X^{\hat\mu,\hat\alpha}}\bigr)
    = 0.
\end{align}
For later use, we also introduce the notation $\mub^n_t := \frac1n \sum_{i=1}^n \delta_{\Xb^i_t}$.
Let $\delta\in(0,\frac1p)$, and choose any sequence $(i_n,x_0^n)_{n\ge1} \subset \{1,\dots,n\}\times B(n^\delta)$ satisfying
\begin{equation} \label{eq:poc_x_0_ineq}
    \sup_{i\le n}
    \sup_{x\in B(n^\delta)}
    w_n(i,x)
    - 2^{-n}
    \le 
    w_n(i_n,x_n),
    ~\mbox{where}~
    w_n(i,x)
    :=
    \Wc_p\Bigl(
        \Lc^{|X_0^i=x}_{X^{i,\hat\alphab}_T},
        \Lc^{|X_0=x}_{X^{\mu,\hat\alpha}_T}
    \Bigr),
\end{equation}
and let us prove that
\begin{align} \label{eq:poc_x_0}
    \lim_{n\to\infty}
    w_n(i_n,x_n)
    = 0,
\end{align}
which implies by the continuity of the reward functional $U$ that
\[
\big\|\overline{\Gbf}^\delta_n[x_0,\boldsymbol{\alpha}^n]
\big\|_{\L^\infty(\S)}
=
\max_{i\le n}
    \sup_{x_0\in B(n^\delta)}
    \Bigl|
        J_i^n(x_0,\alphab^n)
        -
        U\Big(
            \Lc^{|X_0=x_0}_{X^{\hat\mu,\hat\alpha}_T}
        \Big)
    \Bigr|
\underset{n\to\infty}{\longrightarrow} 0.
\]
To see that \eqref{eq:poc_x_0_ineq} holds, we obtain by standard SDE estimates using theLipschitz continuity of $b$ in $(x,m)$ that
\begin{align*}
\E\bigl[
        \Wc_p(\mu^n_t,\mub^n_t)
        \,\big|\, X^{i_n}_0=x_0^n
    \bigr]&\le
    \E\Bigl[
        \sup_{s\le t}\!
        |X^{j,\alphab}_s - \Xb^j_s|
        \,\Big|\, X^{i_n}_0 = x_0^n
    \Bigr]
    \\
    &\le
    C\!\int_0^t\!
        \E\bigl[
            \Wc_p(\mu^n_s,\widehat\mu_s)\,
            \big|\, X^{i_n}_0 = x_0^n
        \bigr]
    \mathrm{d}s
    \\
    &\le
    C\!\int_0^t\!
        \E\bigl[
            \Wc_p(\mu^n_s,\mub^n_s)
            +\Wc_p(\mub^n_s,\widehat\mu_s)
            \,\big|\, X^{i_n}_0=x_0^n
        \bigr]
    \mathrm{d}s,
\end{align*}
By Gronwall's lemma and direct manipulation, this provides
\begin{align} \label{eq:unif_cong_cont}
    \E\!\left[
        \sup_{s\le T}
        |X^{j,\alphab}_s - \Xb^j_s|
        \,\Big|\, X^{i_n}_0 = x_0^n
    \right]
    &\le
    K\!\int_0^T\!
        \E\bigl[
            \Wc_p(\mub^n_s,\widehat\mu_s)
            \,\big|\, X^{i_n}_0=x_0^n
        \bigr]
    \mathrm{d}s 
    \nonumber \\
    &\le
    KT\E\bigl[\Wc_p(\mub^n,\Lc_{X^{\hat\mu,\hat\alpha}})
        \big| X^{i_n}_0=x_0^n
    \bigr],
\end{align}
for some $K>0$ independent of $n$.

As $b$ is bounded, we have $
    \sup_{t\le T} |\Xb^{i_n}_t - W^{i_n}_t|^p
    \le C(|x^n_0|^p + 1)\le C(n^{1-\eps} + 1),$ a.s. and we estimate for all continuous function $\varphi$ with $p$--polynomial growth that
    \[
    \bigl|
        \langle\varphi,\mub^n\rangle
        - \tfrac1n\!\sum_{j\ne i_n}\!\varphi(\Xb^j,X^j_0)
    \bigr|
    =
    \frac1n\,|\varphi(\Xb^{i_n},x_0^n)|
    \le
    \frac{C}{n}\Bigl(
        \sup_{t\le T}|W^{i_n}_t|^p + n^{1-\eps} + 1
    \Bigr)
    \underset{n\to\infty}{\longrightarrow} 0
\]
As $(\Xb^j)_{j\ge1}$ are iid, this implies by the law of large numbers and dominated convergence imply that $\lim_{n\to\infty}
    \E[|
            \langle\varphi,\mub^n\rangle
            - \langle\varphi,\Lc_{X^{\hat\mu,\hat\alpha}}\rangle
        | |X^{i_n}_0=x_0^n]
    = 0$, and by the arbitrariness of $\varphi$, 
\[
    \lim_{n\to\infty}
    \E\bigl[
        \Wc_p(\mub^n,\Lc_{X^{\hat\mu,\hat\alpha}})
        \,\big|\, X^{i_n}_0=x_0^n
    \bigr]
    = 0 .
\]
Substituting in \eqref{eq:unif_cong_cont} and the definition of $w_n$ in \eqref{eq:poc_x_0_ineq} that:
\begin{equation} \label{eq:opti_conv}
\begin{array}{c}
\displaystyle
    \lim_{n\to\infty}
    \E\Big[
        \sup_{t\le T}
        |X^{i_n,\alphab}_t - \Xb^{i_n}_t|
        \Big| X^{i_n}_0=x_0^n
    \Big]
    = 0,
\\
\displaystyle\mbox{and}~
\lim_{n\to\infty}
    \sup_{i\le n}
    \sup_{x\in B(n^\delta)}
    w_n(i,x)
    =
    \lim_{n\to\infty}w_n(i_n,x_0^n)
=0.
\end{array}
\end{equation}

\medskip
\noindent\textit{Step 2: Optimality condition.}  
Let $(\beta^{n},i_n,x_0^n)\in(\Ac_n)^n\times\{1,\ldots,n\}\times B(n^\delta)$ be a $2^{-n}$--maximizer of $\gamma^n
    :=
    \big\| \overline{\Dbf}_n^\delta[\boldsymbol{\alpha}^n]
    \big\|_{\L^\infty(\mu_0)}$:
\begin{align}\label{eq:deltain}
    \gamma^n - 2^{-n}
    \;\le\;
    J_{i_n}^n\bigl(x_0^n,\alphab^{n,(-i_n)}\!\oplus\beta^n\bigr)
    -
    J_{i_n}^n(x_0^n,\alphab^n).
\end{align}
We consider the situation when Player $i$ deviates from the control strategy $\alpha^{i,n}$ to $\beta^n$, while all other players keep their strategies fixed. Denote 
\[
    X^{j,i}
    :=
    X^{j,\,\boldsymbol{\alpha}^{n,(-i)}\oplus\beta^{i,n}},
    ~~
    \mu^{n,i}_t := \frac1n \sum_{k=1}^n \delta_{X^{k,i}_t},
    ~\mbox{and}~
    \mub^n_t := \frac1n \sum_{k=1}^n \delta_{\Xb^k_t},
\]
where $(\Xb^j)_{j\ge1}$ is, as before, the auxiliary mean-field system associated with $(\widehat\mu,\widehat\alpha)$.

By the Lipschitz continuity of $b$ in $(x,m)$, standard SDE stability estimates ensure the existence of a constant 
$C>0$, independent of $n$, such that for any $j\neq i_n$:
\begin{align}
    \E\!\left[
        \sup_{s\le t}
        |X^{j,i_n}_s - \Xb^j_s|
        \;\middle|\;
        X^{i_n}_0 = x_0^n
    \right]
    &\le
    C\!\int_0^t\!
        \E\!\left[
            \Wc_p(\mu^{n,i_n}_s,\widehat\mu_s)
            \;\middle|\;
            X^{i_n}_0 = x_0^n
        \right]\mathrm{d}s
    \nonumber\\
    &\hspace{-40mm}\le
    C\!\int_0^t\!\!
        \Big(
            \E\!\left[
                \Wc_p(\mu^{n,i_n}_s,\mub^n_s)
                \,\middle|\, X^{i_n}_0 = x_0^n
            \right]
            +
            \E\!\left[
                \Wc_p(\mub^n_s,\widehat\mu_s)
                \,\middle|\, X^{i_n}_0 = x_0^n
            \right]
        \Big)
    \mathrm{d}s .
    \label{stabjneqi}
\end{align}
Then
\begin{align*}
\E\!\big[ \Wc_p(\mu^{n,i_n}_t,\mub^n_t)
             \big| X^{i_n}_0 = x_0^n
    \big]
&\le \E\!\Big[ \frac1n \sum_{j\neq i_n} \sup_{s\le t}|X^{j,i_n}_s - \Xb^j_s|
                     + \frac1n\sup_{s\le t}|X^{i_n,i_n}_s - \Xb^{i_n}_s|
                    \big| X^{i_n}_0 = x_0^n
        \Big]
      \\
    &\le
    C\!\int_0^t\!\big(
        \E\!\big[
            \Wc_p(\mu^{n,i_n}_s,\mub^n_s)
            \big| X^{i_n}_0 = x_0^n
        \big]+
        \E\!\big[
            \Wc_p(\mub^n_s,\widehat\mu_s)
            \big| X^{i_n}_0 = x_0^n
        \big]\big)\mathrm{d}s
    \\
    &\quad +
    \frac1n
    \E\!\Big[
        \sup_{s\le t}
        |X^{i_n,i_n}_s - \Xb^{i_n}_s|
        \big| X^{i_n}_0 = x_0^n
    \Big].
\end{align*}
Applying Gronwall’s lemma to the above inequality yields
\begin{align*}
    \E\!\left[
        \Wc_p(\mu^{n,i_n}_t,\mub^n_t)
        \,\middle|\, X^{i_n}_0 = x_0^n
    \right]
    &\le
    C\!\int_0^t\!
        \E\!\left[
            \Wc_p(\mub^n_s,\widehat\mu_s)
            \,\middle|\, X^{i_n}_0 = x_0^n
        \right]\mathrm{d}s
    \\
    &\quad +
    \frac{C}{n}
    \E\!\left[
        \sup_{s\le T}
        |X^{i_n,i_n}_s - \Xb^{i_n}_s|
        \,\middle|\, X^{i_n}_0 = x_0^n
    \right],
\end{align*}
and injecting this bound in \eqref{stabjneqi}, we obtain
\begin{align*}
    \sup_{j\neq i_n}
    \E\!\left[
        \sup_{s\le T}
        |X^{j,i_n}_s - \Xb^j_s|
        \,\middle|\,
        X^{i_n}_0 = x_0^n
    \right]
    &\le
    C\!\int_0^T\!
        \E\!\left[
            \Wc_p(\mub^n_s,\widehat\mu_s)
            \,\middle|\, X^{i_n}_0 = x_0^n
        \right]\mathrm{d}s
    \\
    &\quad +
    \frac{C}{n}
    \E\!\left[
        \sup_{s\le T}
        |X^{i_n,i_n}_s - \Xb^{i_n}_s|
        \,\middle|\, X^{i_n}_0 = x_0^n
    \right].
\end{align*}
As for the deviating player, we get from Step~1 together with the boundedness of $b$ that
\[
    \lim_{n\to\infty}
    \frac1n
    \E\!\Big[
        \sup_{t\le T}
        |X^{i_n,i_n}_t - \Xb^{i_n}_t|
        \big| X^{i_n}_0 = x_0^n
    \Big]
    = 0 .
\]
Combining the preceding estimates with the convergence established in Step~1, we finally deduce that
\begin{align} \label{eq:poc_deviate}
    \lim_{n\to\infty}
    \sup_{t_0\in[0,T]}
    \E\!\big[
        \Wc_p(\mu^{n,i_n}_{t_0},\widehat\mu_{t_0})
        \big| X^{i_n}_0 = x_0^n
    \big]
    = 0,~
    \lim_{n\to\infty}
    \sup_{j\neq i_n}
    \E\!\big[
        \sup_{t\le T}
        |X^{j,i_n}_t - \Xb^j_t|
        \big| X^{i_n}_0 = x_0^n
    \big]
    = 0 .
\end{align}
Let $\Xb^{i,\beta^n}$ be the processes defined by:  
 \begin{align*}
 \Xb^{i,\beta^n}_0=X^i_0~\mbox{and}~
    \mathrm{d}\Xb^{i,\beta^n}_t 
    &= b( t,\Xb^{i,\beta^n}_t,\widehat\mu_t,\beta^{n}_t)\,\mathrm{d}t 
    + \mathrm{d}W^i_t, 
    \quad t \in [0,T],
\end{align*}
By \eqref{eq:opti_conv} and \eqref{eq:poc_deviate}, it follows that,
\begin{align*}
    \lim_{n \to \infty}\Wc_p \left( \Lc^{|X^i_0=x^{n}_0}_{X^{i_n,i_n}_T},\, \Lc^{|X^{i_n}_0=x^{n}_0}_{\Xb^{i_n,\,\beta^{n}}_T} \right)=0\quad \mbox{and}\quad \lim_{n \to \infty}\Wc_p \left( \Lc^{|X^{i_n}_0=x^{n}_0}_{X^{i_n,\alphab}_T},\, \Lc^{|X^{i_n}_0=x^{n}_0}_{\Xb^{i_n,\,\alpha^{i_n,n}}_T} \right)=0.
\end{align*}
We deduce that
\begin{align}
0 \le
    \limsup_{n \to \infty}\gamma^n  
    &\le \limsup_{n \to \infty} J_{i_n}^n\big(x_0^{n},\alphab^{n,(-i_n)}\oplus\beta^{n}\big)
           - J_{i_n}^n(x_0^{n},\alphab^n) 
    \nonumber\\
    & \le   \limsup_{n \to \infty}U \left( \Lc^{|X^{i_n}_0=x^{n}_0}_{\Xb^{i_n,\,\beta^{n}}_T} \right) 
                                             - U \left(  \Lc^{|X^{i_n}_0=x^{n}_0}_{\Xb^{i_n,\,\alpha^{i_n,n}}_T} \right)
   \nonumber\\ &
     =  \limsup_{n \to \infty}U \left( \Lc^{|X^{i_n}_0=x^{n}_0}_{\Xb^{i_n,\,\beta^{n}}_T} \right) 
                                             - U \left(  \Lc^{|X^{i_n}_0=x^{n}_0}_{\Xb^{\hat\mu,\hat\alpha}_T} \right),
 \label{deltanto0:continuous}
\end{align}
as $\Lc^{|X^i_0=x^{i,n}_0}_{\Xb^{i,\,\alpha^{i,n}}_T}=\Lc^{|X_0=x^{i,n}_0}_{\Xb^{\mu,\alpha}_T}$. Moreover, by using the equivalence between weak and strong open--loop formulation, and strong open--loop and  closed--loop formulation in \citeauthor*{el1987compactification} \cite[Theorem 7.3]{el1987compactification}, it follows that
$$
    U \left( \Lc^{|X^{i_n}_0=x^{n}_0}_{\Xb^{i_n,\,\beta^{n}}_T} \right)
    \le
    \sup_{\beta \in \Ac_n}U \left( \Lc^{|X^{i_n}_0=x^{n}_0}_{\Xb^{i_n,\,\beta}_T} \right) 
    =  \sup_{\beta \in \Ac}U \left( \Lc^{|X_0=x^{n}_0}_{\Xb^{\hat\mu,\beta}_T} \right).
$$
Since $\hat\mu$ is a MFG solution associated to $\hat\alpha$, we deduce from \eqref{deltanto0:continuous} that $\lim_{n \to \infty}\delta^n = 0$, thus completing the proof.
\end{proof}

\section{Existence of one-period MFG solutions}
\label{sect:MFGStatic}

{\bf Proof of Proposition \ref{prop_continuous} under Condition (MFG1)} Let us recall that $\mu_0 \in \Pc_q(\S)$ with $q > p \ge 1$.
We first show that $\mu \mapsto J(\mu,\widehat{a})$ is a continuous map from $\Pc_p(\S \x A)\longrightarrow\R$. This is clearly implied by the continuity of  the map $T_0:\Pc_p(\mu_0)\longrightarrow\L^p(\Pc_p(\S),\mu_0)$:  
    $$
        T_0(\mu)=\Lc_{X^{\mu,\hat{a}}_1}^{| X_0},
        ~\mbox{for all}~\mu\in\Pc_p(\mu_0):=\{ \mu \in \Pc_p(\X \x A): \mu\circ X_0^{-1}=\mu_0 \},
    $$
which we now verify. Let $(\mu^n)_{n \ge 1}$ be a sequence converging to $\mu$ in $\Wc_p$. By the continuity condition on $F$ in Assumption ${\rm H}_p$ and the continuity of the optimal response map $\widehat{a}$, we see that 
    \begin{align*}
        \Lc_{X^{\mu^n,\hat{a}}_1}^{| X_0}
        =
        \Lc_{F_{X_{\!0}}^\eta\!(\mu^n, \hat{a}(\mu^n,\Xb_0))}^{| X_0} 
        \underset{n\to\infty}{\longrightarrow}
        \Lc_{F_{X_{\!0}}^\eta\!(\mu, \hat{a}(\mu,\Xb_0))}^{| X_0},
        ~\mbox{a.s. in }\Wc_p.
    \end{align*}
Again by Assumption ${\rm H}_p$, $F$ has affine growth in $\eta$ locally uniformly in $(\mu,a)$. {\color{violet}As $\E[\eta^p]< \infty$,} this implies that $\sup_{n} \E\big[ |X^{\mu^n,\hat{a}}_1|^p\big| X_0\big] < \infty$, a.s. Then, we obtain by the dominated convergence Theorem that
    \begin{align*}
        \lim_{n \to \infty} \E \left[ \Wc_p \big( \Lc_{X^{\mu^n,\hat{a}}_1}^{| X_0}, \Lc_{X^{\mu,\hat{a}}_1}^{| X_0} \big)^p \right]=0,
    \end{align*}
thus justifying the continuity of $T_0$ in $\Pc_p(\mu_0)$. 

\medskip
Now notice that the map $T: \Pc_p(\mu_0) \longrightarrow \Pc_p(\mu_0)$ defined by $T(\mu)= \Lc_{\left(X^{\mu,\hat{a}},\, \hat a (\mu,\Xb_0) \right)}$, inherits the continuity of $T_0$ and $\widehat{a}$. As $\mu_0\in\Pc_q$ with $q>p$ and $\rho_A$ is bounded, it follows that the image $T \left(\Pc_p(\mu_0) \right)$ of $\Pc_p(\mu_0)$ by the application $T$ is a compact subset of the non--empty convex set $\Pc_p(\X \x A)$ and we then deduce from the Schauder fixed point Theorem the existence of a fixed point $\muh$ for $T$. Clearly $\muh$ is a solution of the MFG.
\ep

\medskip\medskip
\noindent {\bf Proof of Proposition \ref{prop_CSahat}} Notice that, the map $\widehat{a}$ satisfies 
    \begin{align*}
        U \big( \Lc_{X^{\hat\mu,\alpha}_1}^{| X_0} \big) 
        \le 
        U \big( \Lc_{X^{\hat\mu,\hat{a}}_1}^{| X_0} \big),
        ~\mbox{a.s. for all}~
        \alpha\in\Ac.
    \end{align*}
    Let $(\mu^n)_{n \ge 1} \subset \Pc_p$ be a sequence of probability measures converging to some $\mu \in \Pc_p$. From our assumption, the sequence $\left(\Lc _{\widehat{a}(\mu^n,x_0,\xi)} \right)_{n \ge 1}$ is  relatively compact in $A$ for a.e. $x_0 \in \S$. There exists a Borel measurable map ${\rm a} (\overline{x}_0): [0,1] \to A$ and a subsequence $(\bar\mu^k:=\mu^{n^{x_0}_k})_{k \ge 1}$ s.t. $\lim_{k \to \infty} \Lc_{\widehat{a}(\bar\mu^k,x_0,\xi)}= \Lc_{\bar{\alpha}(x_0)(\xi)}$ for a.e. $x_0\in\S$. Given our assumptions on $F$, this implies that
    \begin{align*}
        \Lc_{X^{\bar\mu^k,\hat{a}}_1}^{| X_0=x_0}
        = 
        \Lc_{F^\eta_{x_0}(\bar\mu^k, \hat{a}(\bar\mu^k,x_0,\xi))} 
        \longrightarrow 
        \Lc_{F_{x_0}^\eta\left(\mu, {\rm a} (x_0)(\xi) \right)},
        ~\mbox{as}~k \to \infty.
    \end{align*}
    Therefore, for any $\alpha$,
    \begin{align*}
        U \big( \Lc_{X^{\mu,\alpha}_1}^{| X_0=x_0} \big) 
        = \lim_{k \to \infty} U \big( \Lc_{X^{\bar\mu^k,\alpha}_1}^{| X_0=x_0} \big) 
        \le \lim_{k \to \infty}  U \big( \Lc_{X^{\bar\mu^k,\widehat{a}}_1}^{| X_0=x_0} \big)
        = U \big( \Lc_{F^\eta_{x_0}(\mu, {\rm a} (x_0) (\xi))} \big).
    \end{align*}
By the uniqueness of $\widehat{a}$, we deduce that ${\rm a} (x_0)(s)=\widehat{a}(\mu,x_0,s)$ for $\mu_0 \otimes {\rm U}_{[0,1]}-$a.e. $(x_0,s)\in\S \x [0,1]$. We have proved that any subsequence of  $\left( \Lc_{\hat{a}(\mu^n,x_0,\xi)} \right)_{n \ge 1}$ converges to $\Lc_{\hat{a}(\mu,x_0,\xi)}$, and then we have the convergence of the sequence $\left( \Lc_{\hat{a}(\mu^n,x_0,\xi)}\right)_{n \ge 1}$ towards $\Lc_{\hat{a}(\mu,x_0,\xi)}$ for $\mu_0-$a.e. $x_0\in\S$. Finally, we have
\begin{align*}
\lim_{n \to \infty}\Lc_{X^{\mu^{n},\hat{a}}_1}^{| X_0=x_0}
=  \Lc_{F_{x_0}^\eta(\mu, \hat{a}(\mu,x_0,\xi))} 
= \Lc_{X_1^{\mu,\hat{a}}}^{|X_0=x_0},
\end{align*}
which in turn implies the continuity of $\mu \mapsto J(\mu,\widehat{a})$. 
\ep

\begin{remark}{\rm
Assume further that $F$ is jointly continuous in $(x_0,\mu,a)$. Then, by following the same line of argument as in the proof of (ii), we see that, for all $\mu$, the set of continuity points of the optimal response map  $\widehat{a}(\mu,\cdot)$ has  full $\mu_0-$measure.}
\end{remark}

\medskip
\noindent {\bf Proof of Proposition \ref{prop_continuous} under Condition (MFG2)} We handle here the case where an optimal response map $\widehat{\alpha}$ exists but is not necessarily unique. The standard method to address this issue is to introducing the following relaxed formulation of the representative agent problem. Let $(X_0,M)$ the canonical variable on $\S \x \Pc(A)$ and
$$
\bar{\Ac}:=\left\{\pi\in\Pc(\S\x \Pc(A)):\pi\circ X_0^{-1}=\mu_0\;\right\}.
$$
Given a U$_{[0,1]}$ r.v. $\xi$ is   independent of $\eta$, we introduce the relaxed dynamics
$$
 \overline{X}_1^{\mu,M}:=F_{X_0}^\eta\left(\mu,\overline{M}(M,\xi) \right),
$$
where $\overline{M}:\Pc(A) \x [0,1] \to A$ is a Borel map s.t. $m=\Lc_{\overline{M}(m,\xi)}$ for all $m \in \Pc(A)$, and may be chosen continuous, see e.g. \citeauthor*{blackwellDubins83}  \cite{blackwellDubins83}. We denote 
$$
\varphi(\mu,x_0,m):=\Lc_{F^\eta_{x_0}(\mu,\overline{M}(m,\xi))}
~~\mbox{so that}~~
\varphi(\mu,X_0,M) = \Lc_{\overline{X}_1^{\mu,M}}^{|(X_0,M)},~\mbox{a.e.}
$$
and we introduce the relaxed individual optimization problem
\begin{equation}
\overline{V}(\mu)
:=
\sup_{\pi\in\bar\Ac}\overline{J}(\mu,\pi)
:=
\E^\pi\big[U\big(\varphi(\mu,X_0,M)\big)\big].
\end{equation}
We say that $\hat\mu \in \Pc(\X \x A)$ is a solution of the relaxed MFG associated to $\hat\pi\in\bar\Ac$ if:

\medskip
\noindent ${\rm (}\overline{{\rm IO}}{\rm )}$ Given $\hat\mu$, the relaxed strategy $\hat\pi$ solves the individual optimality $\bar J(\hat\mu,\hat\pi)=\bar V(\hat\mu)$, 

\medskip
\noindent ${\rm (}\overline{{\rm FP}}{\rm )}$ and $\hat\mu$ satisfies the fixed point condition 
$$
    \E^{\hat \pi} \left[ \varphi(\hat \mu,X_0,M)(\mathrm{d}x_1) \delta_{(X_0,\;\overline{M}(M,\xi))}(\mathrm{d}x_0,\mathrm{d}a)\right]=\hat\mu(\mathrm{d}x_0,\mathrm{d}x_1,\mathrm{d}a).
$$
\noindent {\bf 1.} Due to the singularity issue of the conditional expectation involved in the value function, we need to introduce one further relaxation the triplet $(X_0,\Lc_{\overline X_1}^{|X_0},\alpha)$ which takes values in the canonical space $\Rc:=\S\x\Pc_p(\S)\x A$. 

\medskip
\noindent {\bf 1.1.} For $\Gamma\in\Pc_p(\Rc)$, the $(X_0,X_1,{\rm a})-$marginal is denoted by 
$$
\mu^\Gamma(\mathrm{d}x_0,\mathrm{d}x_1, \mathrm{d}a)
:=
\E^\Gamma[\phi (\mathrm{d}x_1)\delta_{\left(X_0, \;{\rm a}  \right)}(\mathrm{d}x_0,\mathrm{d}a)],
$$
and we denote an induced regular version of the conditional law by:
$$
\varphi_1(x_0,\Gamma,m):= \varphi(\mu^\Gamma,x_0,m).
$$
We now introduce the set--valued map in $\Phi:\Pc_p(\Rc)\longrightarrow\!\!\!\!\!\to\Pc_p(\Rc)$: 
$$
\Gamma\longmapsto\Phi(\Gamma)
:=
\big\{\pi\circ\big(X_0,\varphi_1(X_0,\Gamma,M),\overline{M}(M,\xi)\big)^{-1}: \pi\in\overline{\Ac}
\big\},
$$
and we observe that the relaxed mean field game problem can be expressed as
\begin{align*}
\overline{V}(\mu^\Gamma)
=
\Vc(\Gamma)
:=
\sup_{\gamma\in \Phi(\Gamma)} \Jc(\Gamma,\gamma),
~\mbox{where}~
\Jc(\Gamma,\gamma)
&:= 
\int_{\S \x \Pc(A)} U \big( \varphi_1(x_0,\Gamma,a) \big) 
                        \pi(\mathrm{d}x_0, \mathrm{d}a) 
\\
& = \int_{\Rc} U (m) \gamma(\mathrm{d}x_0,\mathrm{d}m,\mathrm{d}a).
\end{align*}
Moreover, a relaxed solution of the MFG is now converted into a probability measure 
\begin{equation}\label{MFGrelax2}
\widehat\Gamma\in\Phi(\widehat\Gamma)
~~\mbox{such that}~~
\Vc(\widehat\Gamma)=\Jc(\widehat\Gamma,\widehat\Gamma).
\end{equation}
{\bf 1.2.} We next show that the set--valued map $\Phi$ is continuous: for an arbitrary sequence $(\Gamma^n)_{n \ge 1}\subset\Pc_p(\Rc)$ converging in $\Wc_p$ to $\Gamma\in\Pc_p(\Rc)$,  
\begin{itemize}
\vspace{-2mm}\item defining $\gamma_n:=\pi\circ(X_0,\varphi_1(X_0,\Gamma_n,M),\overline{M}(M,\xi))^{-1}\in\Phi(\Gamma_n)$ for some fixed $\pi\in\widehat\Ac$, it follows from our continuity assumptions on $F$ that 
$$
    \gamma_n\longrightarrow\gamma:=\pi\circ(X_0,\varphi_1(X_0,\Gamma,M),\overline{M}(M,\xi))^{-1}\in\Phi(\Gamma)\,\mbox{in}\,\Wc_p,
$$ 
thus proving the upper hemicontinuity of $\Phi$;
\vspace{-2mm}\item and any $\gamma:=\pi\circ(X_0,\varphi_1(X_0,\Gamma,M),\overline{M}(M,\xi))^{-1}\in\Phi(\Gamma)$ corresponding to some $\pi\in\widehat\Ac$ is the limit in $\Wc_p$ of the above sequence $(\gamma_n)_{n\ge 1}$, thus proving the lower hemicontinuity of $\Phi$.
\end{itemize}  
We finally show that $\Jc$ is a continuous map from $\Gc r(\Phi)\longrightarrow\R$, where $\Gc r(\Phi)$ denotes the graph of $\Phi$. To see this, let $(\Gamma_n)_{n\ge 1}$ and $\gamma_n\in\Phi(\Gamma_n)$, with corresponding $\pi_n\in\widehat{\Ac}$, be sequences converging in $\Wc_p$ to $\Gamma$ and $\gamma$, respectively. Then
$$
\Jc(\Gamma_n,\gamma_n)
=
\int_{\Rc} U (m) \gamma_n (\mathrm{d}x_0,\mathrm{d}m,\mathrm{d}a)
\longrightarrow
\int_{\Rc} U (m) \gamma (\mathrm{d}x_0,\mathrm{d}m,\mathrm{d}a) 
=\Jc(\Gamma,\gamma),
$$
by the continuity of $U$ in $\Wc_p$.

\medskip\noindent
{\bf 1.3.} Consider the ``argmax$"$ correspondence $\Gamma\longrightarrow\!\!\!\!\to\widehat\Phi(\Gamma):=\{\gamma\in\Phi(\Gamma):\Jc(\Gamma,\gamma)=\Vc(\Gamma)\}$ which is obviously convex. Using the growth conditions on $F$ and the fact that $\mu_0 \in \Pc_q(\S)$ with $q >p$, clearly, the set $\Phi(\Gamma) \subset \Pc_q(\Pc_q(\S) \x \S \x A)$ is convex and compact for all $\Gamma \in \Pc_p( \Rc)$. Then, by the continuity properties of  $\Phi$ and $\Jc$ established in the previous step, we may apply the Berge Maximum Theorem \citeauthor*{aliprantis2006infinite} \cite[Theorem 17.31]{aliprantis2006infinite} to conclude that $\widehat\Phi(\Gamma)$ has nonempty compact values. We are now in the context of the Kakutani--Fan--Glicksberg fixed point Theorem \cite[Corollary 17.55]{aliprantis2006infinite}, and we may conclude that the set of fixed points of $\hat\Phi$ is compact and nonempty, hence the existence of a probability measure $\widehat\Gamma$ satisfying \eqref{MFGrelax2}. 

\medskip
\noindent {\bf 2.} We finally induce from the fixed point $\widehat\Gamma$ a solution of the MFG associated with a Borel map $\widehat{a}\in\Ac$.
As $\widehat\Gamma \in \Phi(\widehat\Gamma)$, we have $\widehat\Gamma=\Lc^\pi_{(X_0,\varphi_1(X_0,\hat\Gamma,M),\bar{M}(M,\xi))}$ for some $\pi\in\bar\Ac$ and
    \begin{align*}
        \mu^{\hat\Gamma}(\mathrm{d}x_0,\mathrm{d}x_1,\mathrm{d}a)=\E^\pi\left[\varphi(\mu^{\hat\Gamma},X_0,M)(\mathrm{d}x_1)\delta_{X_0}(\mathrm{d}x_0) \delta_{\bar{M}(M,\xi)}(\mathrm{d}a) \right].
    \end{align*}
By similar techniques as in \citeauthor*{Filippov1962} \cite{Filippov1962} (see also \citeauthor*{Roxin1962} \cite{Roxin1962}) that Condition \eqref{cond:conv} guarantees the existence of a Borel measurable map $\widehat{a}: \Pc(\X \x A) \x \S \x [0,1] \to A$ s.t.
\begin{equation}\label{hatarandom}
\begin{array}{rcl}
    &\E^\pi \left[ \varphi(\mu^{\hat\Gamma},X_0,M) \mid X_0 \right]
    =\E^\pi \left[ \Lc_{F_{X_0}^\eta(\mu^{\hat\Gamma},\bar{M}(M,\xi) )}^{|(X_0,M)}  
                        \big| X_0 \right]
    =  \Lc_{F_{X_0}^\eta(\mu^{\hat\Gamma},\hat{a}(\mu^{\hat\Gamma}, X_0,\xi ))}^{|X_0},
&\\
   & \E^\pi \big[ f^\ell \big(X_0, 
                                       F^\eta_{X_0}\!(\mu^{{\hat\Gamma}},m),m 
                                \big)_{m=\bar{M}(M,\xi)} 
                \big| X_0 \big] 
    = \E^\pi \big[ f^\ell \big(X_0, F^\eta_{X_0}\!(\mu^{\hat\Gamma},a),a 
                                    \big)_{a=\hat{a}( \mu^{\hat\Gamma}, X_0,\xi )}
                       \big| X_0 \big],
&
\end{array}
\end{equation}
for all $1 \le \ell \le L $, and 
    \begin{equation}\label{Urandom}
        \E^{\pi} \left[U \left( \Lc_{F^\eta_{X_0}(\mu^{\hat\Gamma},\bar{M}(M,\xi))}^{|(X_0,M)} \right) \right] 
        \le 
        \E^{\pi} \left[U \left( \Lc_{F^\eta_{X_0}(\mu^{\hat\Gamma},\hat{a}( \mu^{\hat\Gamma}, X_0,\xi ) )}^{|X_0} \right) \right].
    \end{equation}
We now set 
\begin{align*}
    \muh:= \Lc_{(X_0, F^\eta_{X_0}\!(\mu^{\hat\Gamma},\hat a), \hat a)}
    ~~\mbox{with}~~
    \widehat{a}:=\widehat{a}( \mu^{\hat\Gamma}, X_0,\xi ).
\end{align*}
Notice that, for each $\alpha \in \Ac$, we have $F^\eta_{X_0}(\mu^{\widehat\Gamma},\alpha(X_0,\xi))= F^\eta_{X_0}(\muh,\alpha(X_0,\xi))$ a.e., $\muh(\mathrm{d}x,A)=\mu^{\widehat\Gamma}(\mathrm{d}x,A)$ by \eqref{hatarandom}, and $\overline{V}(\mu^{\widehat\Gamma}) \le J(\muh, \widehat{a})$ by \eqref{Urandom}. Direct verification now shows that $\muh$ is a MFG solution associated to $\widehat{a}(\mu^{\widehat\Gamma},X_0,\xi)$.
\ep

\section{Existence of MFG solution in continuous time}
\label{sect:MFGContinuous}

\begin{proof}
Similar to the one-period setting, we handle the potential discontinuities in the conditional law by means of an appropriate relaxed formulation of the control problem, replacing the control process by a measure-valued process thus allowing for compactness and continuity under weak convergence. 

\medskip
\noindent {\bf 1.} We first consider the space of relaxed controls:
\[
    \M(A)
    :=
    \Bigl\{
        m \in \Mc_+(A \times [0,T]) \;:\;
        m(A,\mathrm{d}t)=\mathrm{d}t
        \text{ and }
        m(\mathrm{d}a,[0,T]) \in \Pc(A)
    \Bigr\},
\]
and we introduce the set of relaxed controls $\overline{\Ac}$ as the collection of all $\F$--predictable $\Pc(A)$--valued processes $(\gamma_t)_{t\in[0,T]}$ with $\gamma_t(\mathrm{d}a)\mathrm{d}t \in \M(A)$ almost surely.  
In particular, $\Ac \subset \overline{\Ac}$ by identifying any strict control $\beta \in \Ac$ with the measure valued process $\delta_{\beta(t,X_0,X_t)}(\mathrm{d}a)\mathrm{d}t$.

Given a flow of measures $\mu = (\mu_t)_{t\in[0,T]} \subset \Pc_p(\R^{2d})$ and a relaxed control $\gamma \in \overline{\Ac}$, we define the state process $\overline X^{\mu,\gamma}$ as the unique strong solution of
\[
    \overline X^{\mu,\gamma}_0 = X_0,
    ~~
    \mathrm{d}\overline X^{\mu,\gamma}_t
    =
    \int_A b(t,X_0,\overline X^{\mu,\gamma}_t,\mu_t,a)\,\gamma_t(\mathrm{d}a)\,\mathrm{d}t
    + \mathrm{d}W_t,
    ~~t\in[0,T].
\]
The process $(\overline X,W,\mu)$ takes values in $\Cc\times\Cc\times\M(A)$ with canonical coordinate $(\Xh,\Wh,\widehat{\gamma})$. Our MFG is based on the joint law of the pair $(X_0,\Lc^{|X_0}_{(\overline X,W,\mu)})$ which takes values in  $\Rc := \R^d\times \Pc_p(\Cc \times \Cc \times \M(A))$, equipped with its product topology, and we denote the corresponding canonical process on $\Rc$ by $(X_0,\muh)$. 

We finally introduce the admissible set of probability measures on $\Rc$
\[
    \overline{\Rc}
    :=
    \bigl\{
        \Pr \in \Pc_p(\Rc)
        :\;
        \Pr\circ X_0^{-1}=\eta,
        ~\text{and}~\Wh~\text{is a $\muh(\omega)-$Brownian motion for} 
        ~\Pr-\mbox{a.e.}~\omega
    \bigr\}.
\]
A probability measure in $\Pr\in\Rc$ is said to be compatible with the measure flow $\mu$ if
$$
\Qr-\text{a.e. }\omega \in \Rc,
        \muh(\omega)
        =
        \P\circ
        \bigl(
            \overline X^{\mu,\theta(\omega),X_0(\omega)},
            W,
            \theta(\omega)
        \bigr)^{-1}
        \text{for some }
        \theta(\omega)\in\overline{\Ac}.
$$
For each $\Pr\in\Pc_p(\Rc)$ and $t\in[0,T]$, notice that the corresponding joint law of $(X_t,X_0)$ under $\Pr$ is given by:
\[
    \mu_t^{\Pr}(\mathrm{d}x,\mathrm{d}x_0)
    :=
    \E^{\Pr}\!\left[
        \muh\circ(\Xh_t)^{-1}(\mathrm{d}x)\,
        \delta_{X_0}(\mathrm{d}x_0)
    \right].
\]
We then introduce our relaxed formulation
\begin{equation}\label{relax:cont}
    \Vc(\Pr)
    := \sup_{\Qr\in\Phi(\Pr)} \Jc(\Qr,\Pr),
    ~~
    \Jc(\Qr,\Pr)
    := \E^{\Qr}\!\left[
        U\bigl(\muh\circ\Xh_T^{-1}\bigr)
    \right].
\end{equation}
where the set-valued map $\Phi : \overline{\Rc}
    \longrightarrow
    \!\!\!\!\!\to
    \overline{\Rc}$ is defined by:
\[
    \Phi(\Pr)
    :=
    \{
        \Qr \in \overline{\Rc}:~\Qr-\mbox{compatible with}~\mu^{\Pr}
    \},
    ~\mbox{for all}~
    \Pr\in\overline{\Rc}.
\]
It is straightforward to check that every strict control $\beta \in \Ac$ induces 
\[
    \Lc\Bigl(X_0,
                 \Lc^{|X_0}_{X^{\mu^{\Pr},\beta},W,\delta_{\beta(t,X_0,X_t)}(\mathrm{d}a)\mathrm{d}t}
         \Bigr)
    \in \Phi(\Pr)
    ~\mbox{for all}~
    \Pr\in\overline{\Rc}.
\]
Furthermore, if $(\Qr^n)_{n\ge1} \subset \Phi(\Pr)$ converges to some $\Qr$ in the $\Wc_p$ topology, then $\Qr \in \Phi(\Pr)$ as well.  
This follows from standard weak convergence arguments for controlled martingale problems; see, for instance, \cite[Proposition~4.7]{djete2019general}.  
In particular, the graph of $\Phi$ is closed under $\Wc_p$, which will be crucial for the application of Kakutani’s fixed point theorem in the next step.

\medskip
\noindent{\bf 2.} We now prove that the correspondence 
\(\Phi : \overline{\Rc} \longrightarrow \!\!\!\!\!\to \overline{\Rc}\)
is both lower hemicontinuous and upper hemicontinuous.

Let $(\Pr^n)_{n\ge1} \subset \overline{\Rc}$ be a sequence such that
\(\Pr^n \to \Pr\in \overline{\Rc}\) in $\Wc_p$.
For $x_0 \in \R^d$ and $\gamma=(\gamma_t)_{t\in[0,T]} \in \overline{\Ac}$, it follows from the Lipschitz property of $b$ in $(x,m)$ and a Gronwall argument, there exists a constant $C>0$, independent of $n$, such that
\begin{equation}\label{Phi:cont1}
    \E\!\left[
        \sup_{t\in[0,T]}
        \big|X^{\mu^{\Pr^n},\gamma,x_0}_t
            - X^{\mu^{\Pr},\gamma,x_0}_t
        \big|^p
    \right]
    \le
    C
    \sup_{t\in[0,T]}
    \Wc_p\big(\mu_t^{\Pr^n}, \mu_t^{\Pr}\big).
\end{equation}
Moreover, by standard arguments and continuous dependence of SDEs on the measure parameter, the map $\Pr \in \Pc_p(\Rc) \longmapsto (\mu_t^{\Pr})_{t\in[0,T]}\in C([0,T]; \Pc_p(\R^{2d}))$ is continuous under the $\Wc_p$ topology. Then $\lim_{n\to\infty}
    \sup_{t\in[0,T]}
    \Wc_p(\mu_t^{\Pr^n}, \mu_t^{\Pr}) = 0$, and we deduce from \eqref{Phi:cont1} that
\[
    \lim_{n\to\infty}
    \P \circ (X^{\mu^{\Pr^n},\gamma,x_0},\gamma)^{-1}
    =
    \P \circ (X^{\mu^{\Pr},\gamma,x_0},\gamma)^{-1}
    \quad\text{in }\Wc_p.
\]
Since the control space $A$ is compact, it follows that the set of relaxed controls $\M(A)$ is compact under weak convergence.  
Moreover, the boundedness of $b$ and the integrability condition $\Lc(X_0)\in\Wc_q$ for some $q> p$ imply the tightness of the family $\bigcup_{n\ge1}\Phi(\Pr^n)$ in $\Wc_{p}$.
Consequently, any sequence $(\Qr^n)_{n\ge1}$ with $\Qr^n\in\Phi(\Pr^n)$ admits a convergent subsequence in $\Wc_{p}$, and any limit $\Qr$ of a converging subsequence is in $\Phi(\Pr)$.

Conversely, given $\Qr\in\Phi(\Pr)$, one can construct a sequence $(\Qr^n)_{n\ge1}$ with $\Qr^n\in\Phi(\Pr^n)$ and $\Qr^n\to\Qr$ in $\Wc_{p}$.
Hence, $\Phi$ is both upper and lower hemicontinuous.

\medskip
\noindent {\bf 3.} We next establish existence of a fixed point for the maximizers correspondence:
\[
    \widehat{\Phi}(\Pr)
    :=
    \bigl\{
        \Qr \in \Phi(\Pr)
        : \Jc(\Qr,\Pr) = \Vc(\Pr)
    \bigr\}.
\]
By construction, $\widehat{\Phi}(\Pr)$ is convex, and by the Berge Maximum Theorem
(see \cite[Theorem~17.31]{aliprantis2006infinite}),
$\widehat{\Phi}$ has nonempty, compact values and is upper hemicontinuous. Then, it follows from the Kakutani–Fan–Glicksberg fixed point theorem
(\cite[Corollary~17.55]{aliprantis2006infinite}) to the set valued map
$\widehat{\Phi}$ defined on the convex compact subset $\overline{\Rc}$ of $\Pc_p(\Rc)$ that there exists $\widehat\Pr \in \overline{\Rc}$ such that $\widehat\Pr \in \widehat{\Phi}(\widehat\Pr)$.

\medskip
\noindent {\bf 4.} We finally show that the fixed point $\widehat\Pr$ induces a strict MFG solution. By standard results on relaxed controls (see \cite{el1987compactification}), every relaxed control in $\overline{\Ac}$ can be approximated by a sequence of strict controls in $\Ac$:
\[
    \Vc(\Pr)
    =
    \sup_{\beta \in \Ac}
    J\bigl((\mu_t^{\Pr})_{t\in[0,T]},\,\beta\bigr)
    ~\mbox{for all}~
    \Pr \in \overline{\Rc}.
\]
By definition of $\widehat{\Phi}$, we have for $\widehat\Pr$--a.e.\ $\omega$, the conditional law $\muh(\omega)$ is the distribution of the canonical triple $(\Xh,\Wh,\widehat{\gamma})$ solving
\[
    \mathrm{d}\Xh_t
    =
    \int_A
        b\bigl(t,X_0(\omega),\Xh_t,\mu_t^{\Pr},a\bigr)
        \widehat{\gamma}_t(\mathrm{d}a)\,\mathrm{d}t
    + \mathrm{d}\Wh_t,
    \quad
    \Xh_0 = X_0(\omega),
\]
where $\Wh$ is a $\muh(\omega)$--Brownian motion and $\widehat{\gamma}$ is an $\Pc(A)$--valued predictable process representing the relaxed control. To construct a strict control that induces the same marginal dynamics, define the conditional mean kernels
\[
    K(t,\muh(\omega),x)
    :=
    \E^{\muh(\omega)}\!\bigl[
        \widehat{\gamma}_t
        \,\big|\,
        \Xh_t = x
    \bigr]
    ~\mbox{and}~
    \overline{K}(t,x_0,x)
    :=
    \E^{\Pr}\!\bigl[
        K(t,\muh,x)
        \,\big|\,
        X_0 = x_0
    \bigr].
\]
By construction, $\overline{K}$ is a measurable map from
$[0,T]\times\R^d\times\R^d$ into $\Pc(A)$.
Using standard projection (or mimicking) arguments 
(see, e.g., \cite{gyongy1986mimicking}), we can construct a process 
$Y$ satisfying
\[
    Y_0 = X_0,
    \qquad
    \mathrm{d}Y_t
    =
    \int_A
        b\bigl(t,X_0,Y_t,\mu_t^{\Pr},a\bigr)
        \overline{K}(t,X_0,Y_t)(\mathrm{d}a)\,\mathrm{d}t
    + \mathrm{d}W_t,
\]
such that the joint conditional law satisfies
\[
    \Lc^{\Pr}_{X_0,\E^{\Pr}[\muh \circ \Xh_t^{-1}|X_0]}
     =
    \Lc_{X_0,\Lc^{|X_0}_{Y_t}}.
\]
In words, the process $Y$ reproduces the same conditional marginal distributions as the relaxed process $\Xh$. We next exploit the convexity of the control problem to replace the measure-valued control $\overline{K}$ by a point-valued one.  
Since $A$ is convex and the set $\{b(\cdot,a): a\in A\}$ is convex in $\R^d$, the measurable selection theorem (see, e.g., \cite{el1987compactification}) guarantees the existence of a Borel measurable map ${\rm a} 
    : [0,T]\times\R^{2d}\times\Pc(\R^{2d})
    \longrightarrow A,$
such that
\begin{align*}
    \int_A
        b(t,X_0,Y_t,\mu_t^{\Pr},a)
        \overline{K}(t,X_0,Y_t)(\mathrm{d}a)
        =b (t, X_0,Y_t, \mu^{\Pr}_t, {\rm a}(t,X_0,Y_t,\mu^{\Pr}_t )),
\end{align*}
and consequently
\[
    \Lc^{|X_0= x_0}_{Y_t}
    =
    \Lc^{|X_0= x_0}_{X_t^{\mu^{\Pr},\alpha}},
    ~\mbox{with}~
    \alpha(t,x_0,x) := {\rm a} (t,x_0,x,\mu_t^{\Pr})
    ~\mbox{for all}~x_0.
\]
That is, the strict control $\alpha$ reproduces the same conditional law as the relaxed one. In addition, for each $t \in [0,T]$,
\begin{align*}
    \mu^{\Pr}_t = \E^{\Pr}\!\left[
        \muh\circ(\Xh_t)^{-1}(\mathrm{d}x)\,
        \delta_{X_0}(\mathrm{d}x_0)
    \right] = \E^{\Pr}\!\left[
        \Lc(Y_t\mid X_0)(\mathrm{d}x)\,
        \delta_{X_0}(\mathrm{d}x_0)
    \right]=\Lc(X_t^{\mu^{\Pr},\alpha}, X_0).
\end{align*}
Finally, by conditional Jensen’s inequality and the concavity of $U$,
\begin{align*}
    \E^{\Pr}\!\left[
        U\bigl(\muh\circ \Xh_T^{-1}\bigr)
    \right]
    = 
    \E^{\Pr}\!\left[
        \E^{\Pr}\!\left[
            U\bigl(\muh\circ \Xh_T^{-1}\bigr)
            \,\big|\,
            X_0
        \right]
    \right]
    &\le
    \E^{\Pr}\!\left[
        U\Bigl(
            \E^{\Pr}\!\bigl[
                \muh\circ \Xh_T^{-1}
                \,\big|\,
                X_0
            \bigr]
        \Bigr)
    \right]
    \\
    &=
    \E\!\left[
        U\bigl(\Lc(X_T^{\mu^{\Pr},\alpha}\mid X_0)\bigr)
    \right].
\end{align*}
Hence, the strict control $\alpha$ attains the same performance as the relaxed equilibrium. The optimality property of $\alpha$ follows immediately from that of the relaxed equilibrium $\Pr$.
Hence, the flow of measures $(\Lc_{X_0,X_t^{\mu^{\Pr},\alpha}})_{t\in[0,T]}$
satisfies both the fixed-point and individual optimality conditions, and it is therefore an {\rm MFG} solution associated with the strict control $\alpha$.
\end{proof}

\bibliographystyle{plain}


\bibliography{DjeteTouzi-ArxivVersion}

\end{document}